\documentclass[paper=a4, fontsize=11pt]{scrartcl}
\usepackage[T1]{fontenc}
\usepackage{fourier}
\usepackage[english]{babel}								 
\usepackage{amssymb}
\usepackage{amsmath,amsfonts,amsthm} 
\usepackage[pdftex]{graphicx}	
\usepackage{url}
\usepackage[title,titletoc,page]{appendix} 
\usepackage{listings} 
\usepackage{color}
\usepackage{cite}
\usepackage{caption}
\usepackage{subcaption}
\usepackage[ruled,vlined]{algorithm2e}
\usepackage{float}
\usepackage{booktabs}

\usepackage{comment}

\numberwithin{equation}{section}		
\numberwithin{figure}{section}			
\numberwithin{table}{section}				

\newtheorem{defi}{Definition}[section]
\newtheorem{lem}{Lemma}[section]
\newtheorem{thm}{Theorem}[section]

\newtheorem{remark}{Remark}[section]

\newtheorem{prop}{Proposition}[section]
\newtheorem{assum}{Assumption}[section]

\newcommand{\babs}[1]{\left|{#1}\right|}
\newcommand{\bnorm}[1]{\left|\left|{#1}\right|\right|}

\newcommand{\vect}[1]{\boldsymbol{\mathbf{#1}}}

\title{An Operator Theory for Analyzing the Resolution of Multi-illumination Imaging Modalities
	\thanks{\footnotesize This work was supported in part by the Swiss National Science Foundation grant number
		200021--200307.}}
\author{Ping Liu\thanks{\footnotesize Department of Mathematics, ETH Z\"urich, R\"amistrasse 101, CH-8092 Z\"urich, Switzerland (ping.liu@sam.math.ethz.ch, habib.ammari@math.ethz.ch).} \and Habib Ammari\footnotemark[2]}

\date{}
\begin{document}
	\maketitle
	\begin{abstract}
	 By introducing a new operator theory, we provide a unified mathematical theory for  general source resolution in the multi-illumination imaging problem. Our main idea is to transform multi-illumination imaging into single-snapshot imaging with a new imaging kernel that depends on both the illumination patterns and the point spread function of the imaging system. We therefore prove that the resolution of multi-illumination imaging is approximately determined by the essential cutoff frequency of the new imaging kernel, which is roughly limited  by the sum of the cutoff frequency of the point spread function and the maximum essential frequency in the illumination patterns. 
	 
	 Our theory provides a unified way to estimate the resolution of various existing super-resolution modalities and results in the same estimates as those obtained in experiments. In addition, based on the reformulation of the multi-illumination imaging problem, we also estimate the resolution limits for resolving both complex and positive sources by sparsity-based approaches. We show that the resolution of multi-illumination imaging is approximately determined by the new imaging kernel from our operator theory and better resolution can be realized by sparsity-promoting techniques in practice but only for resolving very sparse sources. This explains experimentally observed phenomena in some sparsity-based super-resolution modalities. 	 
	\end{abstract}

	\vspace{0.5cm}
	
	\vspace{0.2cm}
	\vspace{0.5cm}

\section{Introduction}	
Due to the intrinsic property of wave propagation and diffraction, the spatial resolution in optical imaging was deemed to be limited by the optical diffraction limit for more than a century. Based on the criteria first proposed by Abbe \cite{abbe1873beitrage} and Rayleigh \cite{rayleigh1879xxxi}, this limit is commonly acknowledged to be nearly half of the wavelength and it is widely used to quantify the resolution of conventional optical microscopies. However, in the last two decades, pioneered by several super-resolution techniques \cite{STED, betzig2006imaging,STORM}, a large amount of super-resolution fluorescence microscopies were developed to shatter the diffraction barrier and even frequently achieve a resolution that is dramatically lower than the diffraction limit. For example, the techniques known as STED \cite{STED}, PALM \cite{betzig2006imaging}, STORM \cite{STORM} and dSTORM \cite{heilemann2008subdiffraction} exploit fluorescence to improve the spatial resolution from more than two hundreds nanometers to several tens of nanometers.

A crucial and common feature in the  super-resolution fluorescence microscopies is that multiple patterned fields of light were applied to the sample to manipulate its fluorescence emission and multiple snapshots are taken and processed to extract sub-wavelength features of the sample. Since the snapshots are taken from samples subject to multiple illuminations, we call imaging in this setting multi-illumination imaging, to distinguish it from the imaging from a single snapshot. For imaging from a single snapshot (or a single-illumination), as already demonstrated in \cite{liu2021theorylse,  liu2021mathematicaloned, liu2021mathematicalhighd, batenkov2019super}, the required signal-to-noise ratio is very restrictive when super-resolving $n$ point sources separated by a distance below the diffraction limit. Thus super-resolution is nearly hopeless in this case. This is why  practical super-resolution techniques have developed slowly over the last century, where  multiple illuminations have been rarely utilized. In recent years, the capabilities of super-resolution from a single snapshot have already been established by several mathematical theories \cite{liu2021mathematicaloned, liu2021mathematicalhighd, liu2021theorylse,  batenkov2019super, li2021stable, demanet2015recoverability} and the resolution limits have been explicitly characterized \cite{liu2021theorylse, liu2022mathematicalSR, liu2022mathematicalpositive}, while the super-resolution capability of multi-illumination imaging is not yet well understood.

On the other hand, although the mechanism and resolution of the aforementioned imaging modalities were simple and well explained, such as the down-modulating of high-frequency information in SIM and single molecule localization in STORM, a variety of the perspectives of  understanding  do not uncover the fundamental principle and possibilities
for improving the resolution by using multi-illuminations. There is no mathematical theory to understand all or most of the imaging modalities in a unified way, exhibiting the fundamental principle and performance limit in their resolution improvement. In particular, many new imaging modalities employing the prior knowledge of sparsity \cite{beam2022, zhao2022sparse} have achieved a resolution better than common sense, necessitating an investigation of a mathematical theory for the resolution as well. The fundamental understanding would certainly inspire us to develop new imaging modalities and give us insight into their fundamental limitations.
Therefore, the development of a rigorous and uniform mathematical theory to discover the principle and show the resolution of multi-illumination imaging in a straightforward and simple way is important.

This paper aims to present a unified mathematical theory for understanding the resolution of multi-illumination based super-resolution techniques. In particular, we seek to mathematically explain the resolution improvement of existing multi-illumination imaging approaches and further highlight the possibilities and difficulties in this field.

\subsection{Main contributions}
We first propose an operator theory for analyzing multi-illumination imaging. To be more specific, we define the multi-illumination imaging operator as $A$ and apply its adjoint operator $A^*$ to the measurement $A f$ with $f$ being a general source.  It turns out that $A^*A f$ can be viewed as a conventional imaging from a single snapshot with a specific imaging kernel $G(z, y)$. Therefore, this imaging kernel enables us to analyze the resolution of the multi-illumination imaging by conventional ways. Especially, as one shall see in Sections \ref{section:convolutheory1} and \ref{section:resoluofimgingmodality},  all of the multi-illumination imaging methods have this imaging kernel, despite it is quite hidden in some modalities, such as the SIM and the single molecule fluorescence microscopy.

Based on our operator theory, in Sections \ref{section:stabilitytheoryofimging} and \ref{section:resoluofimgingmodality} we analyze the stability of the reconstruction of the frequency information of the source $f$ and show that our results are in agreement with the experimental results.
While this is consistent with common understanding in the field of multi-illumination imaging, our presentation is more general and does not require specific manipulation of specific measurements and illumination patterns. For example, it explains the resolution of structured illumination microscopy, imaging by translating illumination points, single molecular localization microscopy and decoding based random illumination imaging in the same mathematical framework. 

In addition, this general framework allows us to analyze more aspects of  multi-illumination imaging. As shown in Section \ref{section:encoderdecoder}, by generalizing the above operator analysis to a more general encoding and decoding theory, for a large category of decoding methods, the reconstructed spectral data of $f$ cannot exceed $[-\Omega_{\mathrm{multi}}, \Omega_{\mathrm{multi}}]$, where $\Omega_{\mathrm{multi}} = \Omega_{\mathrm{psf}}+\Omega_{\mathrm{illu}}$ with $ \Omega_{\mathrm{psf}}$ being the cutoff frequency of the point spread function and $\Omega_{\mathrm{illu}}$ the essential maximum frequency in the illumination patterns. This demonstrates the common sense in the field that, without additional prior information, the maximum frequency information in the multi-illumination imaging recovery is limited by the summation of the cutoff frequency of the point spread function and the essential maximum frequency in the illumination pattern.  We also analyze the resolution of multi-illumination imaging for the case when the illumination patterns are not exactly known but can be approximated. 


On the other hand, based on the imaging kernel formulated in Section \ref{section:convolutheory1}, under appropriate assumptions, in Section \ref{section:rslforsparsesr} we are also able to estimate rigorous resolution limits of certain sparsity-based multi-illumination imaging methods. Our results explain important phenomena in some sparsity-based multi-illumination imaging modalities, revealing the inherent advantage and limitation of sparsity-based multi-illumination imaging. In particular, we arrive at the following conclusions: i) the resolution of multi-illumination imaging is fundamentally determined by the summation of the cutoff frequency of the point spread function and the essential maximum frequency in the illumination pattern; and ii) better resolution can be achieved by sparsity-promoting approaches, but only for resolving very sparse sources. 

\subsection{Related works}
The mathematical theory analyzing the ability of super-resolution  imaging from a single noisy snapshot dates back  to the last century. From the middle of the last century, many researchers have already analyzed the two-point resolution from the perspective of statistical inference \cite{helstrom1964detection, helstrom1969detection, lucy1992statistical, lucy1992resolution, den1996model, goodman2015statistical, shahram2004imaging, shahram2004statistical, shahram2005resolvability}. In these papers, the authors have derived estimations for the minimum SNR that is required to discriminate two point sources or for the possibility of a correct decision. Although the resolutions (or the requirement) in this respect were thoroughly explored in these works which spanned the course of several decades, these results are complicated, wherefore they are rarely used in practical applications. Recently, we proposed a new rigorous and simple formula in \cite{liu2022mathematicalSR} to serve as a resolution limit in super-resolving two point sources under only an assumption on the noise level. 

The mathematical analysis of the stability for recovering more than two point sources is more challenging. To the best of our knowledge, the first breakthrough dates back to Donoho. In 1992, he studied the possibility and difficulties of super-resolving multiple on-the-grid sources from  a noisy single snapshot. He derived both the lower and upper bounds for the minimax error of the amplitude recovery in terms of the noise level, grid spacing, cutoff frequency, and a so-called Rayleigh index. The results were improved in recent years for the case when resolving $n$-sparse on-the-grid sources \cite{demanet2015recoverability}. Especially, the authors showed that the minimax error rate for amplitudes recovery scales like $\mathrm{SRF}^{2n-1}\epsilon$, where $\epsilon$ is the noise level and $\mathrm{SRF}:= \frac{1}{\Delta\Omega}$ is the super-resolution factor with $\Delta$ being the grid spacing and $\Omega$ the band limit. Similar results for multi-cluster cases were also derived in \cite{batenkov2020conditioning, li2021stable}. In particular, in \cite{batenkov2019super} the authors derived sharp minimax errors for the location and the amplitude recovery of off-the-grid sources. They showed that for complex sources satisfying a specific clustered configuration and $\epsilon \lesssim \mathrm{SRF}^{-2p+1}$ with $p$ being the number of the cluster nodes, the minimax error rate for reconstructing of the cluster nodes is of  order $(\mathrm{SRF})^{2p-2} \frac{\epsilon}{\Omega}$, while for recovering the corresponding amplitudes the rate is of order $(\mathrm{SRF})^{2p-1}\epsilon$. These results were generalized to the case of superresolving positive sources by us \cite{liu2022super} recently. 

On the other hand, in order to characterize the exact resolution in the number and location recovery, in the earlier works \cite{liu2021mathematicaloned, liu2021theorylse, liu2021mathematicalhighd, liu2022nearly, liu2022mathematicalpositive, liu2022mathematicalSR}, we have defined the so-called "computational resolution limits", which characterize the minimum required distance between point sources so that their number or locations can be stably resolved under certain noise level. It was shown that the computational resolution limits for the number and location recoveries in the $k$-dimensional super-resolution problem should be around respectively $\frac{C_{num}(k, n)}{\Omega}\left(\frac{\sigma}{m_{\min }}\right)^{\frac{1}{2 n-2}}$ and $\frac{C_{supp}(k, n)}{\Omega}\left(\frac{\sigma}{m_{\min }}\right)^{\frac{1}{2 n-1}}$, where $C_{num}(k,n)$ and $C_{supp}(k,n)$ are certain constants depending only on the source number $n$ and the space dimensionality $k$. In particular, these results were generalized to the case when resolving positive sources in \cite{liu2022mathematicalpositive}. We also refer the readers to \cite{moitra2015super,chen2020algorithmic} for understanding the resolution limit from the perspective of sample complexity and to  \cite{tang2015resolution, da2020stable} for the resolving limit of some algorithms.

All of these results reveal the severe ill-conditioning of the inverse problem, indicating that achieving super-resolution for resolving multiple sources from single snapshot is almost hopeless. This is also the reason why the practical super-resolution techniques have developed slowly over the past century. The significant development of practical super-resolution techniques in the last two decades is mainly attributed to the use of multiple illuminations with different patterns. 

Although super-resolution techniques have achieved considerable progress and have become indispensable tools for understanding biological functions at the molecular level, the mathematical theory regarding the possibility and difficulty in multi-illumination imaging has not been satisfactory developed. 
To the best of our knowledge, stability estimations for the MUSIC and ESPRIT algorithms for multi-snapshots have been developed in \cite{li2021stability}. In our previous work \cite{liu2022mathematical}, we proposed a theory for the resolution estimation of sparsity-based multi-illumination imaging, revealing the importance of the incoherence of the illumination patterns in the resolution enhancement. However, the theory did not provide the performance limit of the multi-illumination imaging with known illumination patterns and the results still lack practical significance.

Here we propose a theory focused on the analysis of  multi-illumination imaging with exactly known or well approximated illumination patterns. We also intend our results 
to have sufficient interpretability and guiding significance for practical super-resolution techniques.

\subsection{Organization of the paper}
In Section 2, we first propose an operator analysis for multi-illumination imaging. In particular, we formulate a certain imaging kernel in the multi-illumination case and analyze the stability of the reconstruction of frequency information of a general source. In Section 3, we examine several super-resolution microscopies to elucidate their stability using the operator theory presented in Section 2. In Section 4, we derive estimates of the resolution limit of  sparsity-based multi-illumination imaging. Section 5 concludes the paper. The appendix contains several technical proofs.  
		
		\section{Imaging kernel and resolution of multi-illumination imaging modalities}\label{section:convolutheory1}
		In this section, we propose a mathematical theory to analyze the resolution of multi-illumination imaging. Our theory is based on the analysis of  imaging operators that appear in  multi-illumination imaging problems. For convenience, we call our theory the operator theory for the resolution analysis.

		\subsection{Problem setting and the Imaging kernel}\label{section:imagingkernel1}
		Let us first introduce the problem setting. We suppose that a general source $f$ is supported on $[0,1]^d$ and the point spread function of the imaging system is given by $$k(x,y)=PSF(x-y)$$ with $y$ denoting the source location. We also suppose that we have $N$ times of illuminations for the source and the illumination patterns, denoted by $I(x,t_q), q=1, \cdots, N$,  are known a priori. We consider having full data for each image on $\mathbb R^d$ (or full spectral data in the low-frequency region) to gain more precise reconstruction and make the analysis convenient. Moreover, we make the following assumptions on  $f$, $k(x,y)$ and $I(x,t_q)$. Our assumptions are consistent with the practical modalities. 
		
		\begin{assum}\label{asum:source1}
			$f$ is either a continuous function in $[0,1]^d$ or a discrete and finite measure. 
		\end{assum}
		\begin{assum}\label{asum:psf1}
			$k(x,y)$ is a real, continuous, and bounded function and $PSF\in L^2(\mathbb R^{d})$. 
		\end{assum}
		\begin{assum}\label{asum:illu1}
			$I(x,t_q), q=1, \cdots, N,$ are continuous and bounded functions in $\mathbb R^d$. 
		\end{assum}
		
		The noiseless measurements are given by
		\[
		\hat f(x,t_q)=\int_{\mathbb R^d}k(x,y)I(y,t_q)f(y)dy, \quad q=1, \cdots, N,
		\]
		where $f(y)$ is the unknown source. The multi-illumination imaging problem is to reconstruct $f(x)$ from $\hat f(x,t_q)$. To obtain an appropriate method for analyzing the resolution of the multi-illumination imaging problem, we define the imaging operator $A$ by
		\begin{equation}\label{equ:imagingoperator1}
		Af(x, t_q)=\int_{[0,1]^d}k(x,y)I(y,t_q)f(y)dy=\int_{[0,1]^d}Q(x,y,t_q)f(y)dy,
		\end{equation}
		where the function $Q(x,y,t_q)=k(x,y)I(y,t_q)$ combines together the point spread function and the illumination pattern. Note that by Assumptions \ref{asum:source1}, \ref{asum:psf1} and \ref{asum:illu1}, it is not difficult to see that $Af(x,t_q)\in L^2(\mathbb R^d)$ for each $t_q$. Then we define the inner product 
	\[
	\left\langle Af, g \right \rangle= \frac{1}{N}\sum_{q=1}^N\int_{\mathbb R^d}Af(x,t_q)\overline{g(x,t_q)}dx
	\]
	for $g \in \left(L^2(\mathbb R^d)\right)^N$ and  calculate the adjoint operator $A^*$ of $A$ from
	\begin{align*}
			\left\langle f,A^*g \right \rangle & = \left\langle Af,g \right \rangle=\frac{1}{N}\sum_{q=1}^N\int_{\mathbb R^d} \int_{[0,1]^d}Q(x,y,t_q)f(y)dy \ \overline{g(x,t_q)}dx\\
		&=\int_{[0,1]^d}\left( \frac{1}{N}\sum_{q=1}^N\int_{\mathbb R^d} Q(x,y,t_q)\overline{ g(x,t_q)}dx\right )f(y)dy\\
		&=\left\langle f,\frac{1}{N}  \sum_{q=1}^N\int_{\mathbb R^d}Q(x,y,t_q)\overline{ g(x,t_q)}dx\right \rangle.
	\end{align*}
Thus we get that
	\begin{equation}\label{equ:astar1}
		A^*g =\frac{1}{N}\sum_{q=1}^N \int_{\mathbb R^d}\overline{ Q(x,y,t_q)}g(x,t_q)dx.
	\end{equation}  		
Next, we consider the operator $A^*A$. We have 
\begin{equation}\label{equ:astara1}
		\begin{aligned}
			A^*Af(z)&=\frac{1}{N}\sum_{q=1}^N\int_{\mathbb R^d}\overline{Q(x,z,t_q)}\int_{[0,1]^d}Q(x,y,t_q)f(y)dydx\\
			&=\int_{[0,1]^d} \frac{1}{N}\sum_{q=1}^N \int_{\mathbb R^d}\overline{Q(x,z,t_q)}Q(x,y,t_q)dxf(y)dy\\
			&= \int_{[0,1]^d} G(z,y)f(y)dy,		
		\end{aligned}
	\end{equation}
		where $G(z,y):=\frac{1}{N}\sum_{q=1}^N \int_{\mathbb R^d}\overline{Q(x,z,t_q)}Q(x,y,t_q)dx$.
		Moreover, we obtain that
		\begin{align*}
			G(z,y)&=\frac{1}{N}\sum_{q=1}^N \int_{\mathbb R^d}\overline{Q(x,z,t_q)}Q(x,y,t_q)dx\\
			&=\frac{1}{N}\sum_{q=1}^N \int_{\mathbb R^d}\overline{I(z,t_q)}I(y,t_q)\overline{k(x,z)}k(x,y)dx\\
			&=\frac{1}{N}\sum_{q=1}^N\overline{I(z,t_q)}I(y,t_q)\int_{\mathbb R^d}\overline{k(x,z)}k(x,y)dx.
		\end{align*} 
	Since $PSF\in L^2(\mathbb R^d)$, $\int_{\mathbb R^d}\overline{k(x,z)}k(x,y)dx<+\infty$ and $G(z,y)$ is well-defined. We call
		\begin{equation}\label{equ:imagingkernel1}
			G(z,y):=\frac{1}{N}\sum_{q=1}^N\overline{I(z,t_q)}I(y,t_q)\int_{\mathbb R^d}\overline{k(x,z)}k(x,y)dx, \quad (z,y)\in \mathbb{R}^{2d},
		\end{equation}
	the imaging kernel of the multi-illumination imaging problem. From the above derivations, the problems in multi-illumination imaging are now transformed to conventional imaging problems from single snapshot with the imaging kernel $G(z,y)$. This is a crucial contribution of the paper, which allows us to transform the multi-illumination imaging to a standard imaging from a single measurement, where abundant techniques and results can be applied. In the rest of the paper, we shall see that the imaging kernel $G(z,y)$ plays a key role in determining the stability and resolution of multi-illumination imaging. 		
		
		\subsection{Stability analysis of the Imaging Problem}\label{section:stabilitytheoryofimging}
		In the above section, we have derived an imaging kernel in the multi-illumination imaging problem. We now propose some stability analyses for the multi-illumination imaging, elucidating that the resolution of the multi-illumination imaging is almost determined by the imaging kernel $G$ in (\ref{equ:imagingkernel1}). 
		
		To be more specific, we denote the  noise function  by $\sigma(x,t)$ and suppose that 
		\begin{equation}\label{equ:noiselevel1}
			\bnorm{\sigma(x,t_q)}_{L^1} =\int_{\mathbb R^d} \left|\sigma(x,t_q)\right|dx\leq \sigma, \quad q=1, \cdots, N,
		\end{equation}
		with $\sigma$ being the noise level.  Here we slightly abuse the use of $\sigma$ to keep the notation simple,  but this will not cause any ambiguity in the following discussions. For noisy measurements in the multi-illumination imaging given by
		\[
		h(x,t_q) = Af(x,t_q) +\sigma(x,t_q),\quad q=1, \cdots, N,
		\] 
		with $\sigma(x,t_q)$ satisfying (\ref{equ:noiselevel1}), we have 
		\begin{align*}
			A^*h =& A^*Af+A^*\sigma\\
			=&\int_{[0,1]^d}G(z,y)f(y)dy+A^*\sigma\\
			=&\int_{[0,1]^d}G(z,y)f(y)dy+ O(\sigma),
		\end{align*}
		where the last equality is because by (\ref{equ:noiselevel1}) the following estimate holds: 
		\begin{equation}\label{equ:noiseestimate1}
			\begin{aligned}
				\babs{A^*\sigma(y)}\leq & \frac{1}{N} \sum_{q=1}^N\babs{\int_{\mathbb R^d}\overline{I(y,t_q)k(x,y)}\sigma(x,t_q)dx}\\
				\leq&  \frac{1}{N} \sum_{q=1}^N \max\left(\babs{\overline{I(y,t_q)k(x,y)}}\right) \int_{\mathbb R^d}\babs{\sigma(x,t_q)}dx\\ 
				\leq& C\sigma.
			\end{aligned}
		\end{equation}
	Therefore, by applying the reconstruction operator $A^*$ to the noisy measurement $h(x,t)$, we obtain the following noisy image of the source $f(y)$:
		\begin{equation}\label{equ:imagingequation1}
			\vect Y(z) = \int_{[0,1]^d}G(z,y)f(y)dy+\vect W(z),
		\end{equation}
		where $|\vect W(z)|$ is of order $O(\sigma)$. 
		
		Classically, the resolution of imaging modalities is generally determined by the bandwidth of their point spread functions. Thus, starting from (\ref{equ:imagingequation1})  in the next section (Section \ref{section:resoluofimgingmodality}), we begin to analyze the resolution of various well-known imaging modalities, and even  SIM methods where the imaging kernel is completely concealed. To be more specific, we will show that in several imaging modalities, the imaging kernel $G(z,y)$ is of the form $PSF_{\mathrm{multi}}(z-y)$ with $PSF_{\mathrm{multi}}$ being the new point spread function for the multi-illumination imaging, or $G(z,y)$ can be approximated by $PSF_{\mathrm{multi}}(z-y), \ (z,y)\in [0,1]^{2d},$ to certain extent. This enables us to view (\ref{equ:imagingequation1}) as 
		\begin{equation}\label{equ:imagingequation2}
			\vect Y(z) = \int_{[0,1]^d}PSF_{\mathrm{multi}}(z-y)f(y)dy+\vect W(z)
		\end{equation}
		with $|\vect W(z)|$ being of order $O(\sigma)$. This is exactly the imaging model for the single snapshot, whereby we can understand the resolution in a conventional way, for instance from the point of view of the Rayleigh limit \cite{rayleigh1879xxxi} or the computational resolution limit \cite{liu2021mathematicaloned, liu2021theorylse, liu2021mathematicalhighd, liu2022mathematicalSR}.


	    
	    Furthermore, assuming that $G(z,y)$ is of the form $PSF_{\mathrm{multi}}(z-y)$, we  next demonstrate that the spectral data of $f$ in the bandpass of $PSF_{\mathrm{multi}}$ can be  reconstructed in a stable way from noisy images. In the next subsection, we show that this is also what can only eventually be reconstructed under most circumstances. This provides a rigorous explanation of the resolution of multi-illumination imaging.


		Note that the noisy measurements should be of the form $Af(x,t_q)+\sigma(x,t_q), q=1, \cdots, N,$ for some noise $\sigma(x,t_q)$ satisfying (\ref{equ:noiselevel1}).  Then for $g$ supported in $[0,1]^d$ satisfying $Ag(x,t_q)=Af(x,t_q)+\sigma(x,t_q)$, we have $A(g-f)(x, t_q)=\sigma(x,t_q)$ and $A^*A(g-f)=A^*\sigma$. We define 
		\begin{equation}\label{equ:tildefg1}
		\tilde{f}(y)=
		\left \{
		\begin{array}{cc}
			f(y),   &y \in [0,1]^d,\\
			0,     &y\in \mathbb R^d \setminus [0,1]^d,
		\end{array}
		\right. \quad  \tilde{g}(y)=
		\left \{
		\begin{array}{cc}
			g(y),   &y \in [0,1]^d,\\
			0,    &y \in \mathbb R^d \setminus [0,1]^d.
		\end{array}
		\right.
		\end{equation}
	We denote the Fourier transform by $\mathcal F[\mu](\xi) = \int_{-\infty}^{+\infty}e^{i \xi x} \mu(x) dx$ and consider  
		\begin{align*}
			\mathcal F[A^*\sigma]&=\mathcal{F}[A^*A(g-f)]=\mathcal{F}\left[\int_{\mathbb R^d}G(z,y)(\tilde{g}(y)-\tilde{f}(y))dy\right]\\
			&=\mathcal{F}\left[PSF_{\mathrm{multi}}*(\tilde{g}-\tilde{f})\right]=\mathcal{F}\left[PSF_{\mathrm{multi}}\right]\mathcal F\left[\tilde{g}-\tilde{f}\right]\\
			& = \mathcal{F}\left[PSF_{\mathrm{multi}}\right]\mathcal F\left[g-f\right].
		\end{align*}
		It then follows that
		\[
		\mathcal{F}\left[PSF_{\mathrm{multi}}\right](\xi)\mathcal F\left[g-f\right](\xi) = \mathcal F [A^*\sigma](\xi), \quad \xi \in \mathbb R^d,
		\]
		and,
		\begin{align}\label{equ:frestability1equ0}
			\mathcal F\left[g-f\right](\xi) = \frac{\mathcal F [A^*\sigma]}{\mathcal F\left[PSF_{\mathrm{multi}}\right](\xi)}
		\end{align}
		for $\xi$ satisfying $\mathcal F[PSF_{\mathrm{multi}}](\xi)\neq 0$.  We can therefore reconstruct  the frequency information in the bandpass of $PSF_{\mathrm{multi}}$ in  a stable way. In particular, if we consider $PSF(y)I(y, t_q)\in L^1(\mathbb R^d)$ and assume that $\mathcal F\left[I(y, t_q)\right]$ is bounded, then we have the following theorem for the stability of the reconstruction of the frequency information of $f$. Note that the assumption is very mild and applies to most of the imaging modalities. 
		\begin{thm}\label{thm:frestability1}
		Suppose that $k(x,y)=PSF(x-y)$, $PSF(y)I(y, t_q)\in L^1(\mathbb R^d)$, $\mathcal F\left[I(y, t_q)\right]$'s are bounded, and the imaging kernel in (\ref{equ:imagingkernel1}) has the form $PSF_{\mathrm{multi}}(z-y)$. For $g$ supported on $[0,1]^d$ satisfying $Ag(x,t_q)=Af(x,t_q)+\sigma(x,t_q)$ with $\sigma(x,t_q)$ satisfying (\ref{equ:noiselevel1}), we have 
	  \begin{align*}
	  	\mathcal F\left[g-f\right](\xi) \leq  \frac{C(I,k) \sigma}{\mathcal F\left[PSF_{\mathrm{multi}}\right](\xi)}
	  \end{align*}
  for $\xi$ satisfying $\mathcal F[PSF_{\mathrm{multi}}](\xi)\neq 0$.  Here $C(I,k)$ is a finite constant related to the illumination patterns and point spread function $k(x,y)$. 
		\end{thm}
		\begin{proof}
	Since $PSF(y)I(y, t_q)\in L^1(\mathbb R^d)$, 
	\begin{align*}
	\int_{\mathbb R^d} \left|\sigma(x,t_q)\right| \int_{\mathbb R^d}\left|\overline{I(y, t_q) k(x,y)}\right|dydx< +\infty.
	\end{align*}
Then by Fubini's Theorem, we have 
		\begin{align*}
			\int_{\mathbb R^{2d}}\left|\overline{I(y, t_q) k(x,y)}\sigma(x,t_q)\right|dxdy< +\infty.
		\end{align*}
	Thus $\overline{I(y, t_q) k(x,y)}\sigma(x,t_q)\in L^{1}(\mathbb R^{2d})$ and  Fubini's Theorem yields  
	\begin{align*}
	\left|\mathcal F\left[A^*\sigma\right](\xi)\right| = & \left|\int_{\mathbb R^d} e^{iy\cdot \xi}\frac{1}{N}\sum_{q=1}^N\overline{I(y, t_q)}\int_{\mathbb R^d}\overline{k(x,y)}\sigma(x,t_q)dxdy\right|\\
	=&\left|\frac{1}{N}\sum_{q=1}^N\int_{\mathbb R^d} \sigma(x,t_q) \int_{\mathbb R^d}e^{i y\cdot \xi}\overline{I(y, t_q)}\overline{k(x,y)}dydx\right|\\
	\leq &\frac{1}{N}\sum_{q=1}^N \int_{\mathbb R^d} \babs{\sigma(x,t_q)}\babs{\int_{\mathbb R^d}e^{iy\cdot \xi}\overline{I(y, t_q)}\overline{k(x,y)}dy}dx\\
	=& \frac{1}{N}\sum_{q=1}^N \int_{\mathbb R^d} \babs{\sigma(x,t_q)}\babs{\mathcal F\left[\overline{I(y, t_q)}\right]*  \mathcal F\left[\overline{k(x,y)}\right](\xi)}dx \\
	\leq &  C(I,k)\sigma,
	\end{align*}
	where the last inequality is because (\ref{equ:noiselevel1}), $\mathcal F\left[\overline{I(y, t_q)}\right]$ is bounded, and $k(x,y)=PSF(x-y)$ whose spectral data vanishes outside a bounded interval. Combining the previous estimate with (\ref{equ:frestability1equ0}) proves the theorem. 
	\end{proof}

		According to Theorem \ref{thm:frestability1}, the spectral data of $f$ in the bandpass of $PSF_{\mathrm{multi}}$ can be stably reconstructed if we know the illumination pattern $I$ and the point spread function. Furthermore, as discussed in Sections \ref{section:encoderdecoder} and \ref{section:resoluofimgingmodality}, the essential cutoff frequency of $PSF_{\mathrm{multi}}$ is around the sum of the cutoff frequency of the point spread function and the essential maximum frequency in the illumination pattern.  On the other hand, suppose that the essential cut-off frequency of $PSF_{\mathrm{multi}}$ is $\Omega_{\mathrm{multi}}$, by the classical resolution limit theory \cite{den1997resolution}, the resolution enhancement is around $\frac{\Omega_{\mathrm{multi}}}{\Omega}$.  This is a direct conclusion from our operator theory for the resolution enhancement, which is consistent with experimental results from many imaging modalities in practice. Sometimes additional  prior information further improves the resolution enhancement, of which sparsity is the most common and widely used feature \cite{solomon2019sparcom, zhao2022sparse, blindSIM}. In Section \ref{section:rslforsparsesr}, we will analyze the resolution when resolving sparse sources, enabling an explanation of observed phenomena in some experiments.  
		\begin{remark}
			Note that when $G(z,y)$ is not exactly $PSF_{\mathrm{multi}}(z-y)$ but can be approximated by $PSF_{\mathrm{multi}}$ in some sense, we can obtain an estimate similar to (\ref{equ:frestability1equ0}), indicating that the frequency information in the bandpass of $PSF_{\mathrm{multi}}$ can be reconstructed stably in  multi-illumination imaging. 
 		\end{remark}
 	
 	\begin{remark}
 		We remark that Theorem \ref{thm:frestability1} can also be applied to plane wave illuminations $I(y, t_q) = e^{iz(t_q) \cdot y}$. This means that it can be applied to analyze the resolution of SIM imaging modalities \cite{gustafsson2000surpassing}. 
 	\end{remark}
 
 	\subsection{General Encoding and decoding theory}\label{section:encoderdecoder}
 The formulation of the imaging from the operator $A^*A$ in (\ref{equ:astara1}) can also be viewed as a process of decoding measurements in multi-illumination imaging. To be more specific, the operator  $A^*$ can be viewed as a decoder which decodes the source information from the measurements $Af$ and the decoding patterns are the specific $Q(x,y,t_q) = k(x,y)I(y,t_q)$. In some applications, although the illumination patterns as well as $Q(x,y, t_q)$ are not exactly known, an estimated decoding pattern can be used to reconstruct the image of the sources \cite{gur2010linear, garcia2005synthetic}.  The mathematical formulation is as follows. 
 We use the same notation as in Section \ref{section:imagingkernel1}. Each diffraction limited image in the sequence that is captured by the camera (or the detectors) has the following formulation:
 \[
 \hat f(x, t) = \int_{y} f\left(y \right) I\left(y, t\right) PSF\left(x-y\right) dy,
 \]
 where the  illumination pattern $I(y,t)$ varies in time. For example, in \cite{gur2010linear}, $I(y, t)$ represents the nanoparticles distribution that varies in time according to the Brownian motion of the nanoparticles. Then, the decoding pattern is numerically extracted according to some estimation procedures and the reconstruction $r(x)$ is obtained by as follows:
 $$
 r(x)=\int_t\left[\int_{y} f\left(y \right) I\left(y, t\right) PSF\left(x-y\right) dy\right] \tilde{I}(x, t) d t, 
 $$
 where $\tilde{I}(x, t)$ is the digitally estimated decoding pattern. By different means and assumptions on the model \cite{garcia2005synthetic, gur2010linear}, it was shown that 
 \begin{equation}\label{equ:encodedecodeeq1}
 	\int_t I(y, t)\tilde{I}(x,t)dt = Col(x-y)
 \end{equation}
 for some function $Col$ characterizing the correlations between $I, \tilde{I}$. Then the reconstructed image is 
 \[
 r(x)= \int_y G(x-y)f(y)dy,
 \]
 where $G(x) = Col(x)PSF(x)$ combines the encoding/decoding patterns and the point spread function. The further analysis of the resolution can be derived in the same way as those in the paper. 
 
 On the other hand, the generalization of the reconstruction operator $A^*$ by this encoding and decoding theory gives us a new insight into the stability of  multi-illumination imaging. A simple idea is to create new  general decoders to analyze the possibility of further resolution enhancement compared to Section \ref{section:imagingkernel1}. For simplicity, we always consider that the illumination patterns vary continuously and that each diffraction-limited image is 
 \begin{equation}\label{equ:encodedecodeeq2}
 \hat f(x, t) = Af(x, t) : = \int_{y} k(x, y) I(y, t) f(y) dy. 
 \end{equation}
 Since the only known information is the illumination patterns and the point spread function, we consider a very general decoding pattern  $g_1\left(I(z,t)\right)g_2\left(k(x,z)\right)$ with $g_1, g_2$ being two general functions. 
 
 We define the corresponding decoder  
 \[
 Dg = \int_{t}\int_{x}g_1\left(I(z,t)\right)g_2\left(k(x,z)\right) g(x,t)dxdt.
 \]	
 Then the reconstructed image of source $f$ reads
 \[
 DAf = \int_{t} \int_{x} g_1\left(I(z,t)\right)g_2\left(k(x,z)\right) \int_{y} k(x, y) I(y, t) f(y) dy dx dt,
 \]
 or equivalently,
 \begin{align*}
 	DAf = &  \int_{t} \int_{x} g_1\left(I(z,t)\right)g_2\left(k(x,z)\right) \int_{y} k(x, y) I(y, t) f(y) dy dx dt\\
 	=& \int_{y} \int_t  g_1\left(I(z,t)\right) I(y, t) dt \int_x  g_2\left(k(x,z)\right)  k(x,y)dx f(y)dy.\\
 	=&  \int_{y} \int_t  g_1\left(I(z,t)\right) I(y, t) dt \int_x  g_2\left(PSF(x-z)\right)  PSF(x-y)dx f(y)dy\\
 	=&\int_y G(z, y)f(y)dy,
 \end{align*}
 where $G(z, y) =  \int_t  g_1\left(I(z,t)\right) I(y, t) dt \int_x  g_2\left(PSF(x-z)\right)  PSF(x-y)dx$. Note that the above operations are valid under only very mild assumptions on $g_1, g_2, I, k$. When the illumination patterns are generated by point sources or by other assumptions such as in \cite{garcia2005synthetic}, $I(y, t)$ is of the form $I(y,t)=IP(y-t)$ for some illumination function $IP$. Then 
 \[
 DA f = \int_y PSF_{\mathrm{multi}}(z-y)f(y)dy
 \] 
 and 
 \[
 \mathcal F\left[PSF_{\mathrm{multi}}\right](\xi) = \mathcal F\left[g_1(IP)*IP\right]*\mathcal F\left[g_2(PSF)*PSF\right](\xi), \quad \xi\in \mathbb R^d.  
 \]
 Note that 
 \begin{align*}
 	&\mathcal F\left[g_1(IP)*IP\right](\xi) = \mathcal F\left[g_1(IP)\right](\xi)\mathcal F\left[IP\right](\xi) = \mathcal F\left[g_1(IP)\right](\xi)\mathcal F\left[IP\right](\xi) \chi_{\mathcal F[IP]} , \\
 	&\mathcal F\left[g_2(PSF)*PSF\right](\xi) = \mathcal F\left[g_2(PSF)\right](\xi)\mathcal F\left[PSF\right](\xi)= \mathcal F\left[g_2(PSF)\right](\xi)\mathcal F\left[PSF\right](\xi)\chi_{\mathcal F[PSF]},
 \end{align*}
 where $\chi_g$ is the characteristic function of the support of $g$. Thus the spectral data of $DAf$ is still essentially constrained in $[-\Omega_{\mathrm{multi}}, \Omega_{\mathrm{multi}}]$, where $\Omega_{\mathrm{multi}} = \Omega_{\mathrm{psf}} + \Omega_{\mathrm{illu}}$ with $ \Omega_{\mathrm{psf}}$ being the cutoff frequency of the point spread function and $\Omega_{\mathrm{multi}}$ being the essential maximum frequency in the illumination patterns. This directly reveals that the spectral data of $f$ that can be recovered by any sophisticated recovering (decoding) algorithms with the aforementioned form cannot exceed the bound $\Omega_{\mathrm{psf}} + \Omega_{\mathrm{illu}}$. This theoretically confirms the common sense in the super-resolution field \cite{blindSIM} that,  without further assumption and information on high-order statistics of the illumination patterns such as in \cite{dertinger2009fast, solomon2019sparcom}, the maximum frequency information in  multi-illumination imaging recovery is limited by the sum of the cutoff frequency of the point spread function and the essential maximum frequency in the illumination pattern.  
 
 A crucial observation is that, since in the measurement (\ref{equ:encodedecodeeq2}) the kernel in the integral is $k(x,y)I(y,t)$, together with the discussions above, any sophisticated decoding operator cannot further improve the resolution. Thus, the essential way to further improve the resolution is not to manipulate the decoding operator but to manipulate the measurement (\ref{equ:encodedecodeeq2}), for example  by multiplying it with new functions in the integral kernel in (\ref{equ:encodedecodeeq2}). This  was in fact done in \cite{dertinger2009fast, solomon2019sparcom} under assumptions on high-order statistics of the illumination patterns.
 
 Based on the above discussions, we now have clearer understanding of the possible resolution improvement in multi-illumination imaging. This provides sufficient guidance for the development of super-resolution modalities and algorithms. 
 
 \subsection{Illumination patterns are unknown but can be approximated}
 In many practical applications, the illumination patterns are not exactly known but can be approximated. In this section, we analyze the stability of multi-illumination imaging in this case. Suppose that the original illumination patterns are $I(x, t_q), q=1, \cdots, N,$ and the estimated illumination patterns are $\widehat{I}(x,t_q)= I(x, t_q)+\epsilon(x,t_q)$ with a bounded $\epsilon(x,t_q)$ satisfying
 \begin{align}\label{equ:illuperturblevel1}
 	 \bnorm{\epsilon(x,t_q)}_{L^1}\leq \epsilon.
\end{align}
 
 We define the new imaging operator $\widehat A$ by 
\begin{equation}\label{equ:imagingoperator3}
	\widehat Af=\int_{[0,1]^d}k(x,y) \widehat I(y,t_q)f(y)dy=\int_{[0,1]^d} \widehat Q(x,y,t_q)f(y)dy,
\end{equation}
where $\widehat Q(x,y,t_q)=k(x,y)\widehat I(y,t_q)$.  For the noisy images $h(x,t_q) = Af(x,t_q)+\sigma(x, t_q)$ with $A$ being the original imaging operator (\ref{equ:imagingoperator1}), the source $g$ is recovered by using the following constraint:
\[
\bnorm{\widehat{A} g(x, t_q) - h(x, t_q)}_{L^1}\leq \sigma,\quad q=1,\cdots, N,
\]
or
\[
\widehat{A} g(x, t_q) = h(x, t_q)+ \widehat{\sigma}(x, t_q), \quad q=1,\cdots, N,
\] 
for some $\widehat{\sigma}(x, t_q)$ satisfying (\ref{equ:noiselevel1}). 
 
To analyze the stability of the reconstruction, we further derive the operator $A^* \widehat A$.  Since $\epsilon(x,t_q)$ is bounded, it is not difficult to see that $\widehat Af(x, t_q)\in L^2(\mathbb R^d)$ by Assumptions \ref{asum:source1}, \ref{asum:psf1}, \ref{asum:illu1}. In the same way as for the derivations in  Section \ref{section:imagingkernel1}, we can write that
 	\begin{equation}
 	A^* \widehat A f = \int_{[0,1]^d} \hat G(z,y) f(y)dy,
 \end{equation}  		
 where $\widehat G(z,y)$ is defined by 
 \begin{equation}
\widehat G(z,y):=\frac{1}{N}\sum_{q=1}^N\overline{ \widehat I(z,t_q)}I(y,t_q)\int_{\mathbb R^d}\overline{k(x,z)}k(x,y)dx, \quad (z,y)\in \mathbb{R}^{2d}.  
\end{equation}
 
We now have the following theorem for the stability of recovery of the spectral data of the source $f$. 
\begin{thm}\label{thm:frestability2}
	Suppose that $k(x,y)=PSF(x-y)$, $PSF(y)I(y, t_q)\in L^1(\mathbb R^d)$, the $\mathcal F\left[I(y, t_q)\right]$'s are bounded, and the imaging kernel in (\ref{equ:imagingkernel1}) has the form $PSF_{\mathrm{multi}}(z-y)$. For a bounded $g$ supported on $[0,1]^d$ satisfying $\widehat Ag(x,t_q)=Af(x,t_q)+\sigma(x,t_q)$ with $\sigma(x,t_q)$ satisfying (\ref{equ:noiselevel1}) and $\widehat{A}$ defined by (\ref{equ:imagingoperator3}), we have 
	\begin{align*}
		\mathcal F\left[g-f\right](\xi) \leq  \frac{C_1(I,k) \sigma + C_2(I,k)\epsilon }{\mathcal F\left[PSF_{\mathrm{multi}}\right](\xi)}
	\end{align*}
	for $\xi$ satisfying $\mathcal F[PSF_{\mathrm{multi}}](\xi)\neq 0$.  Here, $C_1(I,k), C_2(I,k)$ are finite constants that depend on the illumination patterns and the point spread function $k(x,y)$. 
\end{thm}
\begin{proof}
	Note first that 
	\begin{align*}
	G_{\epsilon}(z,y):=	\widehat G(z,y)- G(z,y)= \frac{1}{N}\sum_{q=1}^N\overline{ \epsilon(z,t_q)}I(y,t_q)\int_{\mathbb R^d}\overline{k(x,z)}k(x,y)dx.
	\end{align*}
	By the condition on $g$, we have 
		\begin{align*}
		\mathcal F[A^*\sigma]=&\mathcal{F}[A^* \widehat Ag- A^* A f]\\
		=&\mathcal{F}\left[\int_{\mathbb R^d}\widehat{G}(z,y)\tilde{g}(y)- G(z,y)\tilde{f}(y)dy\right] \quad  \bigg( \text{$\tilde{f}, \tilde{g}$ defined in (\ref{equ:tildefg1})} \bigg)\\
		=&\mathcal{F}\left[\int_{\mathbb R^d}{G}(z,y)\tilde{g}(y)- G(z,y)\tilde{f}(y)dy\right] + \mathcal{F}\left[\int_{\mathbb R^d}{G}_{\epsilon}(z,y)\tilde{g}(y)dy \right] \\ 
		=&\mathcal{F}\left[PSF_{\mathrm{multi}}*(\tilde{g}-\tilde{f})\right]+ \mathcal{F}\left[\int_{\mathbb R^d}{G}_{\epsilon}(z,y)\tilde{g}(y)dy \right]  \\
		=&\mathcal{F}\left[PSF_{\mathrm{multi}}\right]\mathcal F\left[\tilde{g}-\tilde{f}\right]+\mathcal{F}\left[\int_{\mathbb R^d}{G}_{\epsilon}(z,y)\tilde{g}(y)dy \right] \\
		=& \mathcal{F}\left[PSF_{\mathrm{multi}}\right]\mathcal F\left[g-f\right]+\mathcal{F}\left[\int_{\mathbb R^d}{G}_{\epsilon}(z,y)\tilde{g}(y)dy \right] .
	\end{align*}
Note that by the proof of Theorem \ref{thm:frestability1}, we have $\babs{\mathcal F[A^*\sigma]} \leq C(I, k)\sigma$ where $C(I, k)$ is a finite constant depending on the illumination pattern and the point spread function. Now we estimate $\mathcal{F}\left[\int_{\mathbb R^d}{G}_{\epsilon}(z,y)\tilde{g}(y)dy \right]$. We have 
\begin{align}\label{equ:prooffrestable2equ1}
	&\babs{\mathcal{F}\left[\int_{\mathbb R^d}{G}_{\epsilon}(z,y)\tilde{g}(y)dy \right](\xi)} \nonumber \\
	=& \babs{\mathcal{F}\left[\int_{\mathbb R^d}\frac{1}{N}\sum_{q=1}^N\overline{ \epsilon(z,t_q)}I(y,t_q)\int_{\mathbb R^d}\overline{k(x,z)}k(x,y)dx\tilde{g}(y)dy \right] (\xi)} \nonumber \\
	=&\babs{\int_{\mathbb R^d} e^{i z \xi} \int_{\mathbb R^d}\frac{1}{N}\sum_{q=1}^N\overline{ \epsilon(z,t_q)}I(y,t_q)\int_{\mathbb R^d}\overline{k(x,z)}k(x,y)dx\tilde{g}(y)dy dz }\nonumber \\
	\leq & \frac{1}{N}\sum_{q=1}^{N} \babs{\int_{\mathbb R^d}\int_{\mathbb R^d}\babs{\overline{ \epsilon(z,t_q)}I(y,t_q)} \int_{\mathbb R^d} \babs{\overline{k(x,z)}k(x,y)}dx \babs{\tilde{g}(y)}dy dz}. 
\end{align}
It is not difficult to see that
\[
\babs{\int_{\mathbb R^d}\int_{\mathbb R^d}\babs{\overline{ \epsilon(z,t_q)}I(y,t_q)} \int_{\mathbb R^d} \babs{\overline{k(x,z)}k(x,y)}dx \babs{\tilde{g}(y)}dy dz}< +\infty.
\]
Thus 
\[
\int_{\mathbb R^{3d}} \babs{\overline{ \epsilon(z,t_q)}I(y,t_q)\overline{k(x,z)}k(x,y)\tilde{g}(y)}dx dy dz <+\infty. 
\]
By Fubini's theorem, we have 
\begin{align*}
&\babs{\mathcal{F}\left[\int_{\mathbb R^d}{G}_{\epsilon}(z,y)\tilde{g}(y)dy \right](\xi)} \\
=& \babs{\frac{1}{N}\sum_{q=1}^{N}\int_{\mathbb R^d}  \int_{\mathbb R^d}\int_{\mathbb R^d} e^{iz \xi}\overline{\epsilon(z,t_q)k(x,z)}dz I(y,t_q)k(x,y)\tilde{g}(y)dxdy}\\
\leq& \frac{1}{N} \sum_{q=1}^{N}\babs{\int_{\mathbb R^d}  \int_{\mathbb R^d}\int_{\mathbb R^d} e^{iz \xi}\overline{\epsilon(z,t_q)k(x,z)}dz I(y,t_q)k(x,y)\tilde{g}(y)dxdy}\\
=& \frac{1}{N} \sum_{q=1}^{N}\babs{\int_{\mathbb R^d}  \int_{\mathbb R^d} \mathcal F\left[\epsilon(z,t_q)\right]* \mathcal F\left[k(x,z)\right](\xi, x) I(y,t_q)k(x,y)\tilde{g}(y) dx dy}.
\end{align*}
By (\ref{equ:illuperturblevel1}), we have $\babs{F\left[\epsilon(z,t_q)\right](\xi)}\leq \epsilon$. Meanwhile, $\mathcal F[k(x,z)](\xi)$ is zero outside a bounded interval. Thus we have $\babs{ \mathcal F[\epsilon(z,t_q)]* \mathcal F[k(x,z)](\xi, x)}\leq C_1 \epsilon$ for some constant $C_1$. Then it is not difficult to see that
\[
\babs{\mathcal{F}\left[\int_{\mathbb R^d}{G}_{\epsilon}(z,y)\tilde{g}(y)dy \right](\xi)} \leq C_2 \epsilon
\]
for a certain constant $C_2$ under the conditions of the theorem and the assumptions made in Section \ref{section:imagingkernel1}. This together with (\ref{equ:prooffrestable2equ1}) completes the proof. 
\end{proof}
	
	Theorem \ref{thm:frestability2} shows that if we stably estimate the illumination patterns, then we can always stably reconstruct the spectral information of the source $f$ inside the bandpass of the $PSF_{\mathrm{multi}}$.  This elucidates the stability of many imaging modalities where the illumination patterns are estimated \cite{beam2022}. 
	
		\subsection{Discrete measurement for each image}\label{section:discreteimaging}
		Since the measurement is taken at some discrete points in real applications, to complete our theory,  we show in this section that we have the same imaging kernel $G(z,y)$ for the case of discrete measurement under certain conditions.  For the sake of presentation, we further make the following simple assumption on the point spread function, which is also compatible with practical applications. 
		
		\begin{assum}\label{asum:psf2}
		The point spread function $PSF$ is smooth and its gradient is bounded. 
		\end{assum}
		
		Suppose we take the measurement at $M^d$ evenly-spaced points $x_j$'s in $[-R,R]^d$ for a large enough $R$ for a single snaptshot. Suppose we have $N$ times of illuminations, the noiseless measurements are 
		\[
		\hat f(x_j,t_q)=\int_{[0,1]^d}k(x_j,y)I(y,t_q)f(y)dy, \quad j=1, \cdots, M^d, \ q=1,\cdots, N.
		\] 
		We define the operator A by
		\begin{equation}\label{equ:discreteoperator1}
		Af=\int_{[0,1]^d}k(x_j,y)I(y,t_q)f(y)dy=\int_{[0,1]^d}Q(x_j,y,t_q)f(y)dy, \quad j=1, \cdots, M^d, \ q=1,\cdots, N,
		\end{equation}
		where the function $Q(x_j,y,t_q)=k(x_j,y)I(y,t_q)$. Now, we define the inner product by
		\[
		\left\langle Af, g(x,t) \right \rangle= \frac{1}{M^dN}\sum_{q=1}^N\sum_{j=1}^{M^d}Af\overline{g(x_j,t_q)}.
		\]
		Calculating the adjoint operator $A^*$ and $A^*A$, we get that 
		\begin{equation}
			A^*g =\frac{1}{M^dN}\sum_{q=1}^N \sum_{j=1}^{M^d}\overline{ Q(x_j,y,t_q)}g(x_j,t_q),
		\end{equation}  
	 and 
		\begin{align*}
			A^*Af(z)=\int_{[0,1]^d} \frac{1}{N}\sum_{q=1}^N\overline{I(z,t_q)}I(y,t_q)\frac{1}{M^d}\sum_{j=1}^{M^d}\overline{k(x_j,z)}k(x_j,y)f(y)dy.
		\end{align*}
	Define $W(z,y)=\frac{1}{N}\sum_{q=1}^N\overline{I(z,t_q)}I(y,t_q)\frac{1}{M^d}\sum_{j=1}^{M^d}\overline{k(x_j,z)}k(x_j,y)$, we have the following lemma. 
	\begin{lem}\label{lem:discreteapprox1}
	For $(z,y)\in [0,1]^{2d}$ and $M$ sufficiently large, we have 
\begin{equation}
	W(z,y)=\frac{1}{N}\sum_{q=1}^N\overline{I(z,t_q)}I(y,t_q)\frac{1}{(2R)^d}\int_{[-R,R]^d}\overline{k(x,z)}k(x,y)dx+ 	C(R,M)
\end{equation}
with $|C(R,M)|\leq \frac{CR}{M}$ for a finite constant $C$. 
	\end{lem}
\begin{proof}
Let $F(x)=\overline{k(x,z)}k(x,y)$ and $\Delta_{x_j}$ be the hypercube $\left[x_{j,1}, x_{j,1}+\frac{2R}{M}\right]\times \left[x_{j,2}, x_{j,2}+\frac{2R}{M}\right]\times \cdots \times \left[x_{j,d}, x_{j,d}+\frac{2R}{M}\right]$. Then
\begin{align*}
	\frac{1}{(2R)^{d}}\int_{[-R,R]^d}\overline{k(x,z)}k(x,y)dx = \frac{1}{(2R)^d}\sum_{j=1}^{M^d}\int_ {\Delta_{x_j}} F(x) dx.
\end{align*}
On the other hand, 
\begin{align*}
&\babs{ \frac{1}{(2R)^d} 	\sum_{j=1}^{M^d}\int_ {\Delta_{x_j}} F(x) dx - \sum_{j=1}^{M^d} F(x_j)}\\
=&\frac{1}{(2R)^d} \babs{ \sum_{j=1}^{M^d}\int_ {\Delta_{x_j}} F(x) dx - \sum_{j=1}^{M^d}\int_ {\Delta_{x_j}} F(x_j) dx}\\
\leq& \frac{1}{(2R)^d} \sum_{j=1}^{M^d}\int_ {\Delta_{x_j}} \babs{F(x)-F(x_j)} dx\\
\leq & \frac{1}{(2R)^d} \max_{\xi \in \mathbb R^d}\babs{\nabla F(\xi)}\sum_{j=1}^{M^d}\int_ {\Delta_{x_j}} \babs{z-x_j} dz\\
 \leq& \frac{1}{(2R)^d}  \max_{\xi \in \mathbb R^d}\babs{\nabla F(\xi)}\sum_{j=1}^{M^d}\int_ {\Delta_{x_j}} \left(z_1-x_{j,1}+\cdots + z_d-x_{j,d}\right) dz \\ 
= & \frac{1}{(2R)^d} \max_{\xi \in \mathbb R^d}\babs{\nabla F(\xi)}\sum_{j=1}^{M^d} \frac{R}{M} \left(\frac{2R}{M}\right)^d\\
= & \max_{\xi \in \mathbb R^d}\babs{\nabla F(\xi)} \frac{R}{M}.
\end{align*}
By Assumptions \ref{asum:psf1} and \ref{asum:psf2}, for all $(z, y)\in [0,1]^{2d}$, $\max_{\xi \in \mathbb R^d}\babs{\nabla F(\xi)}<C<+\infty$ for a uniform contant $C$. This proves the lemma. 
\end{proof}

		Furthermore, we have the following lemma relating $W(z,y)$ to the imaging kernel $G(z,y)$ in \eqref{equ:imagingkernel1}. 
		\begin{lem}\label{lem:approximatekernel1}
			Under the condition that for $(z,y)\in [0,1]^{2d}$, there exist $\alpha, C>0$ such that
			\[
			\babs{\int_{\mathbb R^d \setminus [-R,R]^d}\overline{k(x,z)}k(x,y)dx}\leq\frac{C}{(R-1)^{\alpha}},
			\] 
			we have 
			\[	W(z,y)=\frac{1}{N(2R)^d}\sum_{q=1}^N\overline{I(z,t_q)}I(y,t_q)\int_{\mathbb R^d}\overline{k(x,z)}k(x,y)dx+O\left(\frac{1}{R^d(R-1)^\alpha}\right)+O\left(\frac{R}{M}\right).
			\]
		\end{lem}
	Therefore, when we consider that $M$ is large enough so that $O\left(\frac{R}{M}\right)$ is of at most the same order as $O\left(\frac{1}{R^d(R-1)^{\alpha}}\right)$, we have 
	\[
	W(z,y)=\frac{1}{(2R)^d}G(z,y)+O\left(\frac{1}{R^d(R-1)^{\alpha}}\right), \quad (z,y)\in [0,1]^2,
	\] 
	for $G(z,y)$ being the imaging kernel defined in (\ref{equ:imagingkernel1}). This demonstrates that we will have the same imaging kernel $G(z,y)$ when taking a sufficient number of discrete measurements. 
	
	We remark that the condition 
		\[
	\babs{\int_{\mathbb R^d \setminus [-R,R]^d}\overline{k(x,z)}k(x,y)dx}\leq\frac{C}{(R-1)^{\alpha}}
	\] 
	in the lemma holds for most of the point spread functions in practice, such as $\frac{\sin |x|}{|x|}$ in the one-dimensional space and $\left(\frac{J_1(r)}{r}\right)^2$ in the two-dimensional space, where $J_1$ is the Bessel function of the first kind and order one.

		\section{Resolution study for some imaging modalities}\label{section:resoluofimgingmodality}
		By the operator theory in Section \ref{section:convolutheory1}, we have shown, for example by (\ref{equ:imagingequation1}), that the imaging kernel in the multi-illumination case can be viewed as $G(z,y)$ defined by (\ref{equ:imagingkernel1}). In this section, we analyze the resolution of some multi-illumination imaging modalities by computing the bandwidth of the point spread function $PSF_{\mathrm{multi}}$ from the kernel $G(z,y)$.  
		
		We will show that $G(z,y)$ is equal to or approximated by $PSF_{\mathrm{multi}}(z-y)$, where
		\[
		PSF_{\mathrm{multi}}(z-y)= f_{\mathrm{ILF}}(z-y)f_{\mathrm{PSF}}(z-y)
		\]
		with $f_{\mathrm{ILF}}$ and $f_{\mathrm{PSF}}$ being determined respectively by the illumination pattern $I$ and the point spread function $PSF$ of the imaging system.  The spectral data of $PSF_{\mathrm{multi}}$ is given by
		\[
		\mathcal F[PSF_{\mathrm{multi}}] = \mathcal{F} [f_{\mathrm{ILF}}]*\mathcal{F} [f_{\mathrm{PSF}}]. 
		\]
		This clearly elucidates that  multi-illumination imaging extends the bandwidth of the point spread functions of the imaging system through convoluting them by the illumination patterns and thus increases the resolution. This explains the resolution of many existing imaging modalities in a new and unified way. It is also consistent with the common sense in  multi-illumination imaging that the resolution is determined by the sum of the cutoff frequency of the point spread function and the essential maximum cutoff frequency in the illumination patterns.

		\subsection{Plane Wave Illumination}\label{section:planewaveillumination}
		In structured illumination microscopy (SIM), the sources are illuminated by plane waves with cutoff frequency $\Omega$. In the one-dimensional case, we have $I(y,t) = e^{i\omega(t_q)\Omega y}$, where $\omega(t_q)$ represents the direction of the plane wave. Therefore, $\overline{I(z,t_q)}I(y,t_q)=e^{-i\omega(t_q)\Omega (z-y)}$. In the one-dimensional case, suppose we illuminate the source by two plane waves with opposite directions
		$$I(y,t_1) = e^{i\Omega y}, \quad I(y, t_2) =  e^{-i\Omega y}.$$ We also recall that the point spread function in a single frame is $k(x,y)=PSF(x-y)$. Thus  $G$ in (\ref{equ:imagingkernel1}) is given by
		\begin{align*}
			G(z,y)=&\frac{1}{2}\left(e^{-i\Omega(z-y)}+e^{i\Omega(z-y)}\right)\int_{-\infty}^{\infty}\overline{PSF(x-z)}PSF(x-y)dx\\
			=&\frac{1}{2}\left(e^{-i\Omega(z-y)}+e^{i\Omega(z-y)}\right)\int_{-\infty}^{\infty}\overline{PSF(x-(z-y))}PSF(x)dx \\
			=&: PSF_{\mathrm{multi}}(z-y). 
		\end{align*}
		For example, when $PSF = \frac{\sin  \Omega x}{x}$, the Fourier expansion of $PSF_{\mathrm{multi}}$ is
		\[
		\mathcal{F}[PSF_{\mathrm{multi}}](\xi)=\frac{c}{2}\left(\delta(x-\Omega)+\delta(x+\Omega)\right)* \mathbb{1}_{[-\Omega,\Omega]}(\xi)=\frac{c}{2}\mathbb{1}_{[-2\Omega,2\Omega]}(\xi)
		\]
		for some constant $c$. The band limit is thus doubled; See also Figure \ref{fig:simspectraldata} for an illustration. This demonstrates the well-known twofold resolution improvement of SIM. The above elucidation can be extended to higher dimensions. In fact, for any $PSF\in L^1(\mathbb R^d)$, we can analyze the stability of SIM method by Theorem \ref{thm:frestability1}.  Although our analysis leads to the same result as the one from the frequency explanation of SIM's resolution enhancement, the arguments of the two explanations are actually different. Our new understanding is based on extracting a new basic imaging kernel in the multi-illumination imaging problem, rather than combining all the frequency information from multiple images in SIM. 
		
		\begin{figure}[!h]
			\centering
			\begin{subfigure}[b]{0.48\textwidth}
				\centering
				\includegraphics[width=\textwidth]{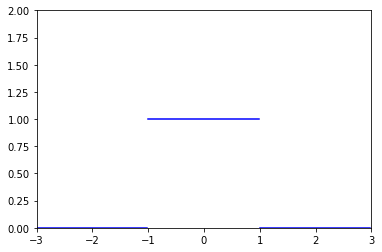}
				\caption{Spectral data of $PSF$}
			\end{subfigure}
			\begin{subfigure}[b]{0.48\textwidth}
				\centering
				\includegraphics[width=\textwidth]{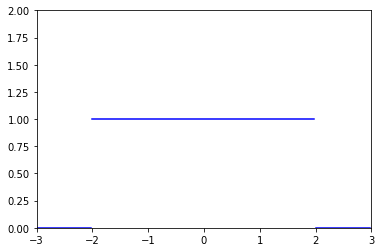}
				\caption{Spectral data of $PSF_{\mathrm{multi}}$}
			\end{subfigure}
			\caption{ Spectral data of $PSF$ and $PSF_{\mathrm{multi}}$ in SIM.}
			\label{fig:simspectraldata}
		\end{figure}

		\subsection{Illumination pattern generated by point sources}\label{section:pointsourceillumination}
		In many super-resolution techniques, the illumination pattern is generated by other sources and has the form of a point spread function. This time, the illumination pattern $I(x,t)$ takes the form $I(x,t)=IP(x-t)$, where $t$ denotes the location of the illumination point. We will show that with a sufficient number of illuminations and measurements, an imaging kernel $G(z,y)$ similar to (\ref{equ:imagingkernel1}) can be derived, which can also be represented by $PSF_{\mathrm{multi}}(z-y)$ with $PSF_{\mathrm{multi}}$ being a new point spread function. For the case when the number of illuminations is not large enough such as in \cite{beam2022}, one would expect that the resolution improvement for resolving general sources will not be better than what derived below. 
		
		To be more specific, we consider the noisy discrete measurements as 
		\begin{equation}\label{equ:noisymeasurement2}
		\hat f(x_j,t_q)=\int_{[0,1]^d}k(x_j,y)I(y,t_q)f(y)dy+\sigma(x_j, t_q), \quad j=1, \cdots, M^d, \ q=1,\cdots, N,
		\end{equation}
		with  $\left|\sigma(x_j, t_q)\right|\leq \sigma$. For convenience of presentation,  we suppose that the illumination points $t_q$'s are evenly spaced in$[-T, T]^d$.  Since $I(x,t)=IP(x-t)$ is generated by point sources, we can assume that
		\begin{align*}
			\babs{\int_{\mathbb R^d \setminus [-T,T]^d}\overline{I(x,z)}I(x,y)dx}\leq\frac{C}{(T-1)^{\alpha}}, \  (z,y)\in [0,1]^{2}. 
			\end{align*} 
		Consequently, in the same way as in Lemma \ref{lem:approximatekernel1}, for large enough $N,M, R, T$, we can have
		\begin{align*}
		&\frac{1}{M^dN}\sum_{q=1}^N \overline{I(z,t_q)}I(y,t_q)\sum_{j=1}^{M^d}\overline{k(x_j,z)}k(x_j,y)\\
		=&\frac{1}{(2T)^d(2R)^d} \int_{\mathbb R^d} \overline{I(z,t)}I(y,t)dt \int_{\mathbb R^d} \overline{k(x,z)}k(x,y) dx+ \frac{1}{(2T)^d(2R)^d} \int_{\mathbb R^d}  O(\sigma). 
		\end{align*}
		Therefore, 
		\begin{align*}
			&\babs{\frac{1}{M^dN}\sum_{q=1}^N \sum_{j=1}^{M^d} \overline{I(z,t_q)} \overline{k(x_j, z)} \sigma(x_j, t_q)}\\
			\leq &\sigma \frac{1}{M^dN}\sum_{q=1}^N \sum_{j=1}^{M^d} \babs{\overline{I(z,t_q)} \overline{k(x_j, z)}} =\frac{1}{(2T)^d(2R)^d} O(\sigma). 
		\end{align*}
		 Thus, by the results of Section \ref{section:discreteimaging}, for the imaging operator $A$ defined in (\ref{equ:discreteoperator1}), we have 
		 \begin{align*}
		 &A^*\hat f =  A^*Af + A^* \sigma \\
		 =& \int_{[0,1]^d} \frac{1}{N}\sum_{q=1}^N\overline{I(z,t_q)}I(y,t_q)\frac{1}{M^d}\sum_{j=1}^{M^d}\overline{k(x_j,z)}k(x_j,y)f(y)dy+\frac{1}{M^dN}\sum_{q=1}^N \sum_{j=1}^{M^d} \overline{I(y,t_q)} \overline{k(x_j, y)} \sigma(x_j, t_q)\\
		 =& \frac{1}{(2T)^d(2R)^d}   \int_{[0,1]^d} \int_{\mathbb{R}^d} \overline{I(z,t)}I(y,t) dt\int_{\mathbb{R}^d}\overline{k(x,z)}k(x,y)dx f(y)dy + O(\sigma). 
		 \end{align*} 
		 Hence, we consider the  imaging kernel $G(z,y)$ as
		 \[
		 G(z,y)=\int_{\mathbb{R}^d} \overline{I(z,t)}I(y,t) dt\int_{\mathbb{R}^d}\overline{k(x,z)}k(x,y)dx,
		 \]
		 and obtain that 
		 \[
		 (2T)^d(2R)^dA^*\hat f = \int_{[0,1]^d} G(z, y)dy + O(\sigma). 
		 \]
		 Since $I(x,t)=IP(x-t)$, we get
		\begin{align}\label{equ:PSFmulti1}
			G(z,y)&=\int_{\mathbb{R}^d} \overline{I(z,t)}I(y,t) dt\int_{\mathbb{R}^d}\overline{k(x,z)}k(x,y)dx \nonumber \\
			&=\int_{\mathbb{R}^d}  \overline{IP(z-t)}IP(y-t)dt \int_{\mathbb{R}^d}\overline{PSF(x-z)}PSF(x-y)dx\nonumber \\
			&=\int_{\mathbb{R}^d} \overline{IP(z-y+t)}IP(t)dt \int_{\mathbb{R}^d}\overline{PSF(x-(z-y))}PSF(x)dx\nonumber \\
			&=:PSF_{\mathrm{multi}}(z-y). 
		\end{align}	   
		Now, for different imaging modalities we can characterize the bandwidth of the corresponding $PSF_{\mathrm{multi}}$'s by their Fourier transforms. 
		
		On the other hand, in many imaging modalities, the sources are illuminated by light generated by multiple sources. For these cases, we can model the illumination patterns $I(y, \tilde{t}_q)$'s by $I(y, \tilde{t}_q) = \sum_{l=1}^L b_lIP(y-\tilde{t}_{q,l})$. Thus, the noisy images are 
		\begin{align*}
		\hat g(x_l, \tilde{t}_q) =& \int_{[0,1]^d}I(y, \tilde{t}_q)k(x_j, y) f(y)dy+\sigma(x_j, \tilde{t}_q) \\
		=& \sum_{l=1}^L b_l \int_{[0,1]^d} IP(y-\tilde{t}_{q,l}) k(x_j, y) f(y)dy+\sigma(x, \tilde{t}_q), \ q=1,\cdots, N, j =1, \cdots, M^d,
		\end{align*}
		with $\sigma(x_j, \tilde{t}_q)$'s being the noise. Since the measurement constraint 
		\[
		\babs{\hat g(x_j, \tilde{t}_q)-  \sum_{l=1}^L b_l \int_{[0,1]^d} IP(y-\tilde{t}_{q,l}) k(x_j, y) f(y)dy}\lesssim \sigma,\ q=1,\cdots, N, j =1, \cdots, M^d,
		\]
		can also be generated by
		\[
		\babs{\hat h(x_j, {t}_q)-  \int_{[0,1]^d} IP(y-{t}_{q}) k(x_j, y) f(y)dy}\lesssim \sigma,\  j =1, \cdots, M^d,
		\]
		for certain $t_q$'s and 
		\[
		\hat h(x_j, {t}_q)=  \int_{[0,1]^d} IP(y-{t}_{q})k(x_j, y) f(y)dy+ \sigma(x_j, {t}_q),
		\] 
		the stability for the imaging where the illumination patterns are generated by multiple sources can be analyzed in the same way as in the case when the illumination patterns are generated by a single point source. 

	As an example, we consider  Brownian excitation amplitude modification (BEAM) developed in \cite{beam2022}. In BEAM, the illumination patterns are the waves scattered by the  randomly moving Ag particles. As shown there, assuming that the Ag particles are far apart from each other, the total field $E_{tot}^{\omega}$ generated by the random array of nanoparticles can then be approximated by 
		\begin{align*}
			E_{tot}^{\omega}(r) \approx E_{inc}^{\omega}(r)+\alpha(\omega) \sum_{j=1}^s G^{\omega}\left(r, r_{j}\right) E_{i n c}^{\omega}\left(r_{j}\right)+\alpha(\omega)^{2} \sum_{j=1}^s \sum_{k=1}^s G^{\omega}\left(r, r_{j}\right) G^{\omega}\left(r_{j}, r_{k}\right) E_{i n c}^{\omega}\left(r_{j}\right),
		\end{align*}
	where $E_{inc}^{\omega}$ is the incident wave, $\alpha(\omega)$ is the polarizability of a single nanoparticle and $G^{\omega}(r, r_j)$ is the Green function. The Green function  $G^{\omega}(z,y)$ is approximately $\frac{e^{ik|z-y|}}{|z-y|}C$ for some constant $C$ with $k$ being the wave number of background medium.  The sources are illuminated by the intensity $I(r)$ of the total field that $I(r)=\left|E_{\text {tot }}^{\omega}(r)\right|^{2}$.
	Since the wave number of the illumination patterns is the wave number of the background medium, which is close to that of the point spread functions. Thus, as we have seen in this section, with multiple illuminations, the cutoff frequency of the point spread function  can be at most increased by about two times, which allows for a corresponding improvement in resolution when resolving general sources.  We note that, from Supplementary Figure 1:E in \cite{beam2022},  BEAM can improve the resolution by more than $1.5$ for both two- and four-source recovery, which is consistent with our theoretical prediction in this paper.

  \subsection{Single molecule  localization microscopy}
 Single molecule localization microscopy (SMLM) \cite{lelek2021single, betzig2006imaging, hess2006ultra, STORM, heilemann2008subdiffraction} describes a family of powerful imaging techniques that dramatically improve spatial resolution to the nanometer scale by computationally localizing individual fluorescent molecules, among which the most well-known imaging modalities are STORM \cite{STORM} and PALM \cite{betzig2006imaging}. 

  Since at each frame, only one point or well-separated point sources are illuminated in such SMLM techniques, we can model the illumination patterns in these imaging modalities as $\delta$ function or a continuous function $IP$ with a sharp peak. Note that by Theorem \ref{thm:frestability1}, we have a stability results for these imaging modalities when the illumination patterns are modeled by a function $IP\in C(\mathbb{R}) \cap L^1(\mathbb{R})$.  Thus, by (\ref{equ:PSFmulti1}), 
  \[
  PSF_{\mathrm{multi}}(z-y) = \int_{\mathbb{R}} \overline{IP(z-y+t)}IP(t)dt \int_{\mathbb{R}}\overline{PSF(x-(z-y))}PSF(x)dx. 
  \]
  Since $IP$ has a sharp peak, the bandwidth of $PSF_{\mathrm{multi}}$ is considerably extended, which ensures that these imaging modalities have excellent resolution improvement. In experiments,  PALM \cite{betzig2006imaging} can improve the resolution by more than ten times.

		\section{Resolution limits for sparsity-based super-resolution}\label{section:rslforsparsesr}
		 We have now shown that the frequency information of a general source $f$ that can be stably reconstructed is in the bandpass of the new point spread function $P S F_{\text {multi }}$. However, when we have a prior information that the source is a collection of point sources, we are able to reconstruct more frequency information by sparsity-promoting algorithms, as demonstrated in the single measurement case \cite{candes2013super, candes2014towards}. On the other hand, experimental evidence \cite{beam2022, zhao2022sparse} has shown that sparsity-promoting approaches can achieve better resolution improvement than discussed in Sections \ref{section:convolutheory1} and \ref{section:resoluofimgingmodality} when resolving very sparse sources. 
		 For example, when resolving two sources, BEAM \cite{beam2022} can achieve a threefold resolution improvement, which is better than twofold improvement discussed in Section \ref{section:resoluofimgingmodality}. In this section, we theoretically estimate the resolution limits of sparsity-promoting approaches in multi-illumination imaging. 
		 In particular, we analyze the resolution limits for the recovery of locations and number of complex and positive point sources. Our conclusion is that it is possible to obtain better resolution than that predicted by  operator theory in Sections \ref{section:convolutheory1} and \ref{section:resoluofimgingmodality}, but only for very sparse sources with high signal-to-noise ratio. For resolving more point sources that are tightly-spaced, the resolution should be the  one predicted by operator theory.  This sheds light on the cause of some experimental phenomena in BEAM. To be specific, as shown in Supplementary Figure 1:E in \cite{beam2022}, the algorithm achieves a threefold resolution improvement when resolving two positive point sources, but fails when resolving six positive point sources.
		
		\subsection{Resolution limit for the location recovery}
		Based on the discussions in Sections \ref{section:convolutheory1} and \ref{section:resoluofimgingmodality}, we have shown that, for some multi-illumination imaging modalities, by introducing a certain recovering operator, we can recover the image of a source $\mu=\sum_{j=1}^n a_j \delta_{y_j}$ from
		\begin{equation}\label{equ:imagingmodelrsl1}
		\vect Y(z) = \int_{[0,1]^d}PSF_{\mathrm{multi}}(z-y)\mu(y)dy+\vect W(z), \quad z\in [0,1]^d,
		\end{equation}
		with $|\vect W(z)|$ being of order $O(\sigma)$. The imaging process is then actually a deconvolution and is similar to the one we considered in \cite{liu2021mathematicaloned}. A simple idea is to compare the resolution enhancement by the theory in this paper, but the results in \cite{liu2021mathematicaloned} can only compare the resolution of imaging modalities with point spread functions of the same shape, such as $PSF_1 = f(\Omega x) $ and $PSF_2= f(2\Omega x)$. It is not possible to explicitly compare the resolution of point spread functions with different shapes.  For example, it is known that imaging with the point spread function $PSF=\left(\frac{\sin \Omega x}{x}\right)^5$ has better resolution than  imaging with $PSF=\left(\frac{\sin \Omega x}{x}\right)$, but the theory in \cite{liu2021mathematicaloned} cannot directly show the difference in resolution between the two cases. This is also a common difficulty in comparing resolutions of deconvolution problems. To circumvent this problem, it is useful and reliable to understand the resolution by measurements in the spatial-frequency domain. From (\ref{equ:imagingmodelrsl1}), considering the parallel model in the spatial-frequency domain, we have 
		\begin{equation}\label{equ:imagingmodelrsl2}
		\vect \Psi(\xi) = \mathcal F[PSF_{\mathrm{multi}}](\xi) \mathcal{F}[\mu](\xi) + \mathbf{\hat W}(\xi), \quad \xi \in \mathbb R^d,
		\end{equation}
		with $\mathbf{\hat W}(\xi)$ of order $O(\sigma)$. Note that by Theorems \ref{thm:frestability1} and \ref{thm:frestability2}, the frequency information of the source $\mu$ that can be stably reconstructed is in the bandpass of $PSF_{\mathrm{multi}}$. Thus, the above model is essential in  multi-illumination imaging, even when the source $\mu$ is not recovered from the deconvolution problem (\ref{equ:imagingmodelrsl1}). For ease of analysis, we assume that $\babs{\mathbf{\hat W}(\xi)}\leq \sigma$ with $\sigma$ being the noise level. The inverse problem consists in  reconstructing $\mu$ from the measurement $\vect G(\xi), \xi \in \mathbb R^d$. 
		
		In order to analyze the resolution, we introduce the following $\sigma$-admissible measures, which cannot be distinguished from the underlying sources without additional prior information.

		\begin{defi} 
		Given the measurement $\vect \Psi(\xi), \xi \in \mathbb R^d$ in (\ref{equ:imagingmodelrsl2}), we say that $\hat{\mu}=\sum_{j=1}^m \hat{a}_j \delta_{\hat{{y}}_j}, \hat{{y}}_j \in \mathbb{R}^d$ is a $\sigma$-admissible discrete measure of $\vect \Psi$ if
		\begin{equation}\label{equ:sparsityresolutionlimiteq-1}
		\babs{\mathcal F[PSF_{\mathrm{multi}}](\xi) \mathcal{F} \left[\hat{\mu}\right]({\xi})- \vect \Psi({\xi})}<\sigma, \quad \forall \xi \in \mathbb{R}^d.
		\end{equation}
		\end{defi}

		For a $b_{\mathrm{lower}}\gg \sigma$, denote $\hat \Omega_{\mathrm{multi, b_{\mathrm{lower}}}}$ by 
		\begin{equation}\label{equ:sparsityresolutionlimiteq-2}
		\hat \Omega_{\mathrm{multi, b_{\mathrm{lower}}}}: = \max \left\{ r >0 : \babs{PSF_{\mathrm{multi}}(\xi)}> b_{\mathrm{lower}}\ \text{ for }\ \bnorm{\xi}_2\leq r,\ \xi \in \mathbb R^d\right\}, 
		\end{equation}
		which can be viewed as the essential cutoff frequency of $PSF_{\mathrm{multi}}$. From (\ref{equ:sparsityresolutionlimiteq-1}), for the $\sigma$-admissible measure $\hat \mu$, we have
		\begin{equation}\label{equ:sparsityresolutionlimiteq0}
		\babs{\mathcal{F} \left[\hat{\mu}\right]({\xi})- \frac{\vect \Psi({\xi})}{\mathcal F[PSF_{\mathrm{multi}}](\xi)}}<\frac{\sigma}{b_{\mathrm{lower}}}, \quad  \bnorm{\xi}_2 \leq  	\hat \Omega_{\mathrm{multi, b_{\mathrm{lower}}}}.
		\end{equation}
		The above model is the same as the one in \cite{liu2021mathematicalhighd} for the single measurement case. Thus, we can obtain a similar theorem for the resolution estimate. Defining 
		\[
		B^d({x}):=\left\{{y} : {y} \in \mathbb{R}^d,\bnorm{{y}-{x}}_2< \frac{n-1}{2 \hat \Omega_{\mathrm{multi}, b_{\mathrm{lower}}}} \right\},
		\]	
		by Theorem 2.7 in \cite{liu2021mathematicalhighd}, we have the following theorem.
		\begin{thm}\label{thm:sparsityresolutionlimit1}
			Let $n \geq 2$, assume that $\mu=\sum_{j=1}^{n} a_{j} \delta_{y_{j}}, y_j \in \mathbb R^d, \min_{j=1,\cdots, n}|a_j|\geq m_{\min},$ is supported on $B^d(\mathbf{0})$ and that
			\begin{equation}\label{equ:sparsityresolutionlimiteq1}
				d_{\min}:=\min _{p \neq j}\bnorm{y_{p}-y_{j}}_2 \geq \frac{C_{supp}(d,n)}{\hat \Omega_{\mathrm{multi}, b_{\mathrm{lower}}}}\left(\frac{\sigma}{m_{\min}b_{\mathrm{lower}}}\right)^{\frac{1}{2 n-1}} ,
			\end{equation}
			for $\hat \Omega_{\mathrm{multi}, b_{\mathrm{lower}}}$ defined in (\ref{equ:sparsityresolutionlimiteq-2}) and a numerical constant $C_{supp}(d,n)$ depending only on $d, n$.  For any $\hat{\mu}=\sum_{j=1}^{n} \hat{a}_{j} \delta_{\hat{y}_{j}}$ supported on $B^d(\mathbf{0})$ and satisfying 
			\[
			\babs{\mathcal F[PSF_{\mathrm{multi}}](\xi) \mathcal{F} \left[\hat{\mu}\right]({\xi})- \vect \Psi({\xi})}<\sigma, \quad \forall \xi \in \mathbb{R}^k,
			\]
			for $\vect \Psi(\xi)$ defined in (\ref{equ:imagingmodelrsl2}), after reordering the $\hat{y}_{j}$ 's, we have
			\[
			\bnorm{\hat{y}_j -y_j}_2<\frac{d_{\min}}{2}, 
			\]
			and 
			$$
			\bnorm{\hat{y}_{j}-y_{j}}_2 \leq \frac{C(d, n)}{\Omega} S R F^{2 n-2} \frac{\sigma}{m_{\min }}, \quad 1 \leq j \leq n,
			$$
			with $C(d, n)$ being a numerical constant depending only on $d,n$.
		\end{thm}

		 By Theorem \ref{thm:sparsityresolutionlimit1}, we show that for very sparse sources, better resolution  than the Rayleigh limit $\frac{c(d)\pi}{\hat \Omega_{\mathrm{multi}, b_{\mathrm{lower}}}}$ can be obtained when the signal-to-noise ratio is sufficiently high. This explains why better resolution than that predicted by operator theory when performing sparsity-promoting recovery in  multi-illumination imaging is attained in the experiments \cite{beam2022, zhao2022sparse}. 
		
		On the other hand, by deriving the following lower bound on the resolution in the worst-case scenario, we can also show that achieving better resolution is very hard when resolving more than two sources.  
		\begin{prop} \label{prop:resolutionofsparsitymultiillu1}
		Let $b_{\mathrm{upper}}:=\max_{\xi \in \mathbb R^d} \babs{\mathcal F[PSF_{\mathrm{multi}}](\xi)}<+\infty$.  For given $0<\sigma<m_{\min}b_{\mathrm{upper}}$ and integer $n \geqslant 2$, let
		\begin{equation}\label{equ:sparsityresolutionlimiteq2}
		\tau=\frac{e^{-1}}{\check \Omega_{\mathrm{multi}, \frac{(n-1)!n!\sigma}{(2n)! m_{\min}}}} \left(\frac{\sigma}{m_{\min}b_{\mathrm{upper}}}\right)^{\frac{1}{2 n-1}},
		\end{equation}
		where $\check \Omega_{\mathrm{multi}, \frac{(n-1)!n!\sigma}{(2n)! m_{\min}}}:= \min \left\{ r >0 : \babs{PSF_{\mathrm{multi}}(\xi)}< \frac{(n-1)!n!\sigma}{(2n)! m_{\min}} \ \text{for}\ \bnorm{\xi}_2\geq r, \xi \in \mathbb R^d\right\} $.  Then, there exist $a$ measure $\mu=\sum_{j=1}^n a_j \delta_{{y}_j}, {y}_j \in \mathbb{R}^k$ with $n$ supports at 
		\[
		\left\{\left(-\frac{\tau}{2}, 0, \ldots, 0\right), \ \left(-\frac{3\tau}{2}, 0, \ldots, 0\right), \ \ldots,\ \left(-\left(n-\frac{1}{2}\right)\tau, 0, \ldots, 0\right)\right\}
		\] 
		and a measure $\hat{\mu}=\sum_{j=1}^n \hat{a}_j \delta_{\hat{\hat{y}}_j}$ with $n$ supports at 
		\[
		\left\{\left(\frac{\tau}{2},0, \ldots, 0\right),\ \left(\frac{3\tau}{2}, 0, \ldots, 0\right),\ \ldots,\ \left(\left(n-\frac{1}{2}\right) \tau, 0, \ldots, 0\right)\right\}
		\] 
		such that 
		\[
		\babs{\mathcal F\left[PSF_{\mathrm{multi}}\right](\xi) \mathcal{F} \left[\hat{\mu}\right]({\xi})- \mathcal F\left[PSF_{\mathrm{multi}}\right](\xi) \mathcal{F}\left[ \mu\right]({\xi})}< \sigma,\  \xi \in \mathbb R^d,\quad \min _{1 \leqslant j \leqslant n}\left|a_j\right|=m_{\min}.
		\] 
		\end{prop}
	\begin{proof}
		See Appendix \ref{section:proofofsparsityresolution}.
	\end{proof}

     Note that the source locations in  the eventually reconstructed measure $\hat \mu$ are completely different and distant from the locations of the underlying sources.  Stable recovery of the source locations in this case is impossible by sparsity-based multi-illumination imaging. Since the distance in (\ref{equ:sparsityresolutionlimiteq2}) deteriorates rapidly as $n$ increases, achieving a resolution less than $\frac{e^{-1}}{\check \Omega_{\mathrm{multi}, \frac{(n-1)!n!\sigma}{(2n)! m_{\min}}}}$ is nearly impossible for recovering sources that are not that sparse. This gives a rigorous proof of the resolution limit of the sparsity-promoting super-resolution when the illumination patterns are known.  
     
     We remark that for many illumination patterns,  $\hat \Omega_{\mathrm{multi}, b_{\mathrm{lower}}}$ and $\check \Omega_{\mathrm{multi}, \frac{(n-1)!n!\sigma}{(2n)! m_{\min}}}$ are close to the sum of the cutoff frequency of the $PSF$ and the essential maximum frequency in the illumination patterns. Also, $b_{\mathrm{upper}}$ and $b_{\mathrm{lower}}$ are comparable. Thus, the resolution limit for the multi-illumination imaging is of order $O\left(\frac{1}{\Omega_{\mathrm{multi}}}\left(\frac{\sigma}{m_{\min}}\right)^{\frac{1}{2n-1}}\right)$,  where $\Omega_{\mathrm{multi}}$ is the essential cutoff frequency of the new point spread function $PSF_{\mathrm{multi}}$. This estimate now helps us to understand the resolution in the sparsity-based multi-illumination imaging with known illumination patterns. It indicates that $\frac{c}{\Omega_{\mathrm{multi}}}$ for some constant $c$ is still the essential resolution in multi-illumination imaging, and thus, a better resolution can be achieved but only for recovering very sparse sources.

    \subsection{Resolution limit for resolving positive sources} 
      Based on the discussions and techniques presented in \cite{liu2022mathematicalpositive}, we can directly generalize the above estimates  of the resolution limit to the super-resolution of positive sources. In particular, we define the positive discrete measure by 
      \[
      \mu = \sum_{j=1}^n a_j \delta_{{y}_j}, \quad a_j >0,
      \]
       and  have the following results. 
       
      		\begin{thm}\label{thm:positivesparsityresolutionlimit1}
      		Let $n \geq 2$.  Assume that $\mu=\sum_{j=1}^{n} a_{j} \delta_{y_{j}}, y_j \in \mathbb R^d, a_j >0,  \min_{j=1,\cdots, n}|a_j|\geq m_{\min},$ is supported on $B^d(\mathbf{0})$ and that
      		\begin{equation}\label{equ:positivesparsityresolutionlimiteq1}
      			d_{\min}:=\min _{p \neq j}\bnorm{y_{p}-y_{j}}_2 \geq \frac{C_{supp}(d,n)}{\hat \Omega_{\mathrm{multi}, b_{\mathrm{lower}}}}\left(\frac{\sigma}{m_{\min}b_{\mathrm{lower}}}\right)^{\frac{1}{2 n-1}} ,
      		\end{equation}
      		for $\hat \Omega_{\mathrm{multi}, b_{\mathrm{lower}}}$ defined in (\ref{equ:sparsityresolutionlimiteq-2}) and a numerical constant $C_{supp}(d,n)$ depending only on $d, n$.  For any positive measure $\hat{\mu}=\sum_{j=1}^{n} \hat{a}_{j} \delta_{\hat{y}_{j}}, \hat a_j >0,$ supported on $B^d(\mathbf{0})$ and satisfying 
      		\[
      		\babs{\mathcal F[PSF_{\mathrm{multi}}](\xi) \mathcal{F} \left[\hat{\mu}\right]({\xi})- \vect \Psi ({\xi})}<\sigma, \quad \forall \xi \in \mathbb{R}^d,
      		\]
      		for $\vect \Psi(\xi)$ defined in (\ref{equ:imagingmodelrsl2}), after reordering the $\hat{y}_{j}$ 's, we have
      		\[
      		\bnorm{\hat{y}_j -y_j}_2<\frac{d_{\min}}{2}, 
      		\]
      		and 
      		$$
      		\bnorm{\hat{y}_{j}-y_{j}}_2 \leq \frac{C(d, n)}{\Omega} S R F^{2 n-2} \frac{\sigma}{m_{\min }}, \quad 1 \leq j \leq n,
      		$$
      		with $C(d, n)$ being a numerical constant depending only on $d,n$.
      	\end{thm}
      
      	Theorem \ref{thm:positivesparsityresolutionlimit1} is a direct consequence of Theorem \ref{thm:sparsityresolutionlimit1}, which shows the possibility of achieving a better resolution for resolving positive sparse sources by sparsity-based multi-illumination imaging.

         \begin{prop} \label{prop:positiveresolutionofsparsitymultiillu1}
     	Let $b_{\mathrm{upper}}:=\max_{\xi \in \mathbb R^d} \babs{\mathcal F[PSF_{\mathrm{multi}}](\xi)}<+\infty$.  For given $0<\sigma<m_{\min}b_{\mathrm{upper}}$ and integer $n \geqslant 2$, let
     	\begin{equation}\label{equ:positivesparsityresolutionlimiteq2}
     		\tau=\frac{e^{-1}}{\check \Omega_{\mathrm{multi}, \frac{(n-1)!n!\sigma}{(2n)! m_{\min}}}} \left(\frac{\sigma}{m_{\min}b_{\mathrm{upper}}}\right)^{\frac{1}{2 n-1}},
     	\end{equation}
     	where $\check \Omega_{\mathrm{multi}, \frac{(n-1)!n!\sigma}{(2n)! m_{\min}}}:= \min \left\{ r >0 : \babs{PSF_{\mathrm{multi}}(\xi)}< \frac{(n-1)!n!\sigma}{(2n)! m_{\min}} \ \text{for}\ \bnorm{\xi}_2\geq r, \xi \in \mathbb R^d\right\} $.  Then there exist $a$ measure $\mu=\sum_{j=1}^n a_j \delta_{{y}_j}, {y}_j \in \mathbb{R}^k$ with $n$ supports at 
     	\[
     	\left\{\left(-\left(n-\frac{3}{2}\right)\tau, 0, \ldots, 0\right),\ \left(-\left(n-\frac{7}{2}\right)\tau, 0, \ldots, 0\right),\ \ldots, \ \left(\left(n-\frac{1}{2}\right)\tau, 0, \ldots, 0\right)\right\}
     	\] 
     	and a measure $\hat{\mu}=\sum_{j=1}^n \hat{a}_j \delta_{\hat{y}_j}$ with $n$ supports at 
     	\[
     	\left\{\left(-\left(n-\frac{1}{2}\right)\tau, 0, \ldots, 0\right),\ \left(-\left(n-\frac{5}{2}\right)\tau, 0, \ldots, 0\right), \ \ldots, \ \left(\left(n-\frac{3}{2}\right) \tau, 0, \ldots, 0\right)\right\}
     	\]
     	such that 
     	\[
     	\babs{\mathcal F[PSF_{\mathrm{multi}}](\xi) \mathcal{F} [\hat{\mu}]({\xi})- \mathcal F[PSF_{\mathrm{multi}}](\xi) \mathcal{F}[ \mu]({\xi})}< \sigma,\  \xi \in \mathbb R^d, 
     	\] 
     	and
     	\[
     	 \quad \min _{1 \leqslant j \leqslant n}\left|a_j\right|=m_{\min}.
     	\] 
     \end{prop}
     \begin{proof}
     	See Appendix \ref{section:proofofsparsityresolution}.
     \end{proof}
 
  Note that, in the example of Proposition \ref{prop:positiveresolutionofsparsitymultiillu1}, it is hard to say which $\hat y_j$ in $\hat \mu$ is the recovered locations of some $y_j$ in the source $\mu$. Thus, Proposition \ref{prop:positiveresolutionofsparsitymultiillu1} provides an upper bound estimate on the resolution enhancement by the sparsity-based multi-illumination imaging. To further demonstrate the instability of the location recovery under this order of separation distance, we state the following proposition. 
  
     \begin{prop} \label{prop:positiveresolutionofsparsitymultiillu2}
     Let $b_{\mathrm{upper}}:=\max_{\xi \in \mathbb R^d} \babs{\mathcal F[PSF_{\mathrm{multi}}](\xi)}<+\infty$.	 For given $0<\sigma<m_{\min}b_{\mathrm{upper}}$ and integer $n \geq 2$, let
     \[
     \tau=\frac{0.2 e^{-1}}{\check \Omega_{\mathrm{multi},  \frac{\pi^2\sigma}{2ne^{11} s^2(n+1)^{10} 2^{2 n-8} m_{\min}}} s^{\frac{2 n+1}{2 n-1}}}\left(\frac{\sigma}{m_{\min }b_{\mathrm{upper}}}\right)^{\frac{1}{2 n-1}} , 
     \]
     where 
     \[
     \check \Omega_{\mathrm{multi},  \frac{\pi^2\sigma}{2ne^{11} s^2(n+1)^{10} 2^{2 n-8} m_{\min}}}:= \min \left\{ r >0 : \babs{PSF_{\mathrm{multi}}(\xi)}< \frac{\pi^2\sigma}{2ne^{11} s^2(n+1)^{10} 2^{2 n-8} m_{\min}} \ \text{for}\ \bnorm{\xi}_2\geq r, \xi \in \mathbb R^d\right\}.
     \]  
     Then there exist a positive measure $\mu=\sum_{j=1}^n a_j \delta_{y_j}$ with $n$ supports at 
    \[
    \left\{t_j=-\frac{s n-2}{2} \tau+\frac{(j-2) s}{2} \tau, \quad  j= 2,4, \cdots, 2 n\right \}
    \] 
    and a positive measure $\hat{\mu}=\sum_{j=1}^n \hat{a}_j \delta_{\hat{y}_j}$ with $n$ supports at 
    \[
    \left\{t_j=t_{4\lceil\frac{j+1}{4}\rceil-2}+(-1)^{\frac{j+1}{2}} \tau, \quad  j=1,3,5, \cdots, 2 n-1\right\}
    \] 
    such that
     \[
    	\babs{\mathcal F[PSF_{\mathrm{multi}}](\xi) \mathcal{F} [\hat{\mu}]({\xi})- \mathcal F[PSF_{\mathrm{multi}}](\xi) \mathcal{F}[ \mu]({\xi})}< \sigma,\ \xi\in \mathbb R^d, \quad \min _{1 \leq j \leq n}\left|a_j\right|=m_{\min } .
     \]
     \end{prop}
           \begin{proof}
     	See Appendix \ref{section:proofofsparsityresolution}.
     \end{proof}
     
     The $n$ underlying sources in $\mu$ in Proposition \ref{prop:positiveresolutionofsparsitymultiillu2} are spaced by
     \[
     s \tau=\frac{0.4 e^{-1}}{\check \Omega_{\mathrm{multi},  \frac{\pi^2\sigma}{2ne^{11} s^2(n+1)^{10} 2^{2 n-8} m_{\min}}} s^{\frac{2}{2 n-1}}}\left(\frac{\sigma}{m_{\min}b_{\mathrm{upper}}}\right)^{\frac{1}{2 n-1}} .
     \]
  It is revealed that when the $n$ point sources are separated by $\frac{c}{\check \Omega_{\mathrm{multi},  \frac{\pi^2\sigma}{2ne^{11} s^2(n+1)^{10} 2^{2 n-8} m_{\min}}}}\left(\frac{\sigma}{m_{\min}b_{\mathrm{upper}}}\right)^{\frac{1}{2 n-1}}$ for some constant $c$, the recovered source locations from the positive $\sigma$-admissible measures can be very unstable; See Figure 2.1 for an illustration.
     
     \begin{figure}[!h]
     	\includegraphics[width=0.7\textwidth]{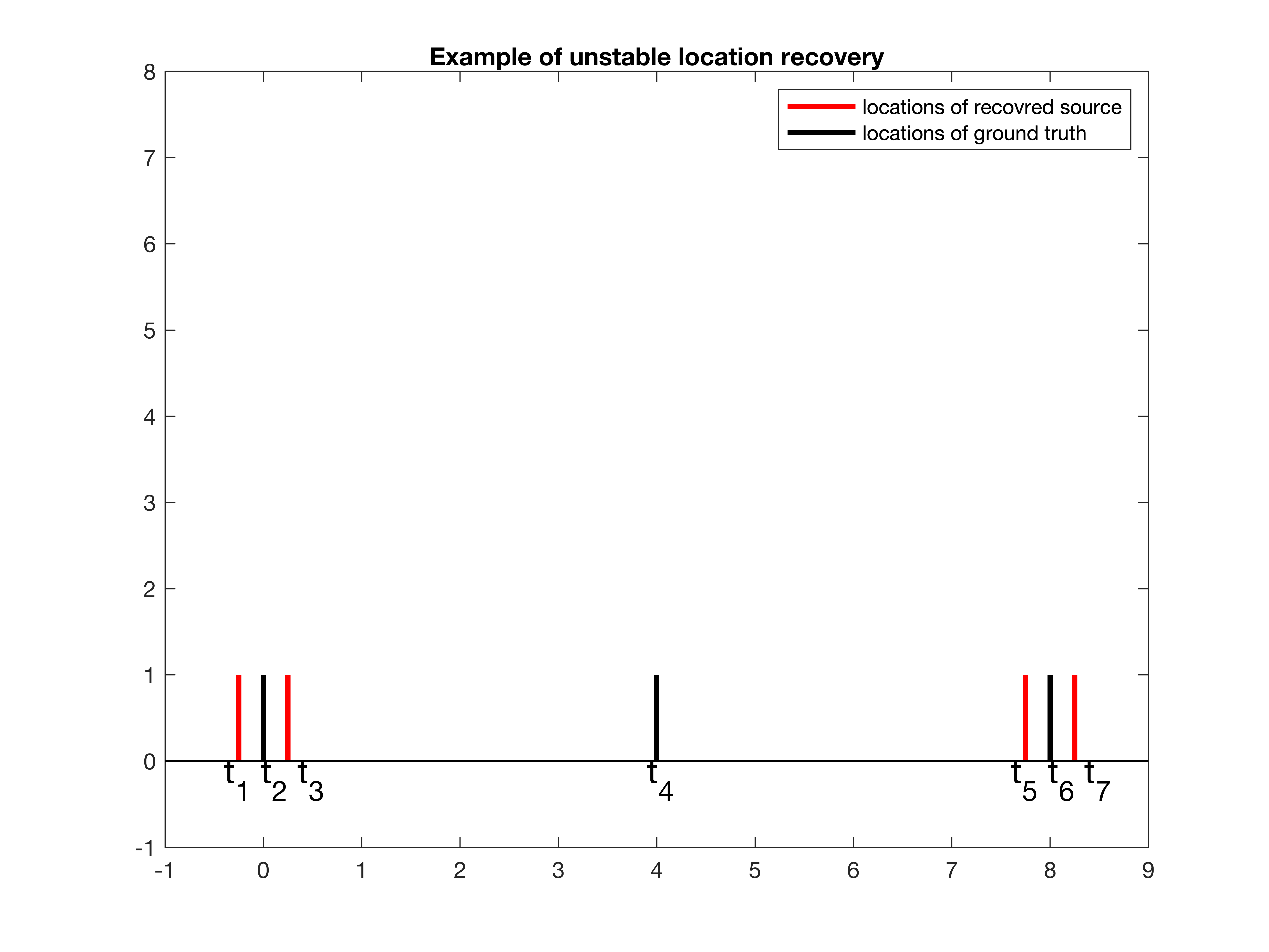}
     	\centering
     	\caption{An example of unstable location recovery. Black spikes indicate the locations of underlying sources and red spikes indicate the source locations of some $\sigma$-admissible measure.}
     	\label{fig:examplelocation}
     \end{figure}

  \subsection{Resolution limit for the source number recovery}
  To better understand the possibilities and difficulties of super-resolution in sparsity-based multi-illumination imaging, in particular the difficulty to achieve a better resolution than that predicted in Section \ref{section:convolutheory1}, we next derive stability results for the recovery of the source number in multi-illumination imaging. We have the following results. 
  
  \begin{thm}\label{thm:numbersparsityresolutionlimit1}
  	Let $n \geq 2$. Assume that $\mu=\sum_{j=1}^{n} a_{j} \delta_{y_{j}}, y_j \in \mathbb R^d, \min_{j=1,\cdots, n}|a_j|\geq m_{\min},$ is supported on $B^d(\mathbf{0})$ and that
  	\begin{equation}\label{equ:numbersparsityresolutionlimiteq1}
  		d_{\min}:=\min _{p \neq j}\bnorm{y_{p}-y_{j}}_2 \geq \frac{C_{num}(d,n)}{\hat \Omega_{\mathrm{multi}, b_{\mathrm{lower}}}}\left(\frac{\sigma}{m_{\min}b_{\mathrm{lower}}}\right)^{\frac{1}{2 n-2}} ,
  	\end{equation}
  	for $\hat \Omega_{\mathrm{multi}, b_{\mathrm{lower}}}$ defined in (\ref{equ:sparsityresolutionlimiteq-2}) and a numerical constant $C_{num}(d,n)$ depending only on $d, n$.  Then there do not exist any measures with less than $n$ supports, $\hat \mu = \sum_{j=1}^k \hat a_j \delta_{\hat {y}_j}, k<n,$ such that 
  	\[
  	\babs{\mathcal F[PSF_{\mathrm{multi}}](\xi) \mathcal{F} \left[\hat{\mu}\right]({\xi})- \vect \Psi({\xi})}<\sigma, \quad \forall \xi \in \mathbb{R}^d.
  	\]
  	In particular, the above results still hold for the case when $\mu, \hat \mu$ are positive measures. 
  	\end{thm}
  \begin{proof}
  	Since the model (\ref{equ:sparsityresolutionlimiteq0}) is the same as the one in \cite{liu2021mathematicalhighd} for the single measurement case, by Theorem 2.3 in \cite{liu2021mathematicalhighd}, we directly obtain the desired results. 
  \end{proof}

Theorem \ref{thm:numbersparsityresolutionlimit1} reveals that when the signal-to-noise ratio is sufficiently high and the source is very sparse, detecting the correct source number when the sources are separated by a distance below $\frac{c}{\hat \Omega_{\mathrm{multi}, b_{\mathrm{lower}}}}$ is possible by sparsity-based multi-illumination imaging. This shows that better resolution can be achieved. However, as the following proposition indicates, in theory this is only possible when resolving very sparse sources. 
  
  \begin{prop} \label{prop:numberresolutionofsparsitymultiillu1} 
  	For given $0<\sigma<m_{\min }$ and integer $n \geq 2$, there exist $\mu=\sum_{j=1}^n a_j \delta_{{y}_j}, {y}_j \in \mathbb{R}^d$ with $n$ supports, and $\hat{\mu}=\sum_{j=1}^{n-1} \hat{a}_j \delta_{\hat{y}_j}$ with $n-1$ supports such that 
  	\[
  	\babs{\mathcal F\left[PSF_{\mathrm{multi}}\right](\xi)\mathcal{F} \left[\hat{\mu}\right](\xi)-\mathcal F\left[PSF_{\mathrm{multi}}\right](\xi)\mathcal{F} \left[\mu\right](\xi)}< \sigma.
  	\] 
  Moreover,
  \[
  \min_{1 \leq j \leq n}\left|a_j\right|=m_{\min }, \quad \min _{p \neq j}\left\|{y}_p-{y}_j\right\|_2=\frac{2 e^{-1}}{\check \Omega_{\mathrm{multi}, \frac{((n-1)!)^2\sigma}{(2n-1)! m_{\min}}}}\left(\frac{\sigma}{m_{\min}b_{\mathrm{upper}}}\right)^{\frac{1}{2 n-2}}, 
  \]
  where $\check \Omega_{\mathrm{multi}, \frac{((n-1)!)^2\sigma}{(2n-1)! m_{\min}}}:= \min \left\{ r >0 : \babs{PSF_{\mathrm{multi}}(\xi)}< \frac{((n-1)!)^2\sigma}{(2n-1)! m_{\min}} \ \text{for}\ \bnorm{\xi}_2\geq r, \xi \in \mathbb R^d\right\} $. 
  In particular, the above results still hold for the case when $\mu, \hat \mu$ are positive measures. 
  \end{prop}
  \begin{proof}
See Appendix \ref{section:proofofsparsityresolution}.
  	\end{proof}
  
  Proposition \ref{prop:numberresolutionofsparsitymultiillu1} demonstrates the challenge of using sparsity-based multi-illumination imaging to super-resolve the number of complex or positive sources. We remark that by discussions in \cite{liu2022mathematicalSR}, it seems that the minimum separation distance in the proposition can actually be 
  \[
  \min _{p \neq j}\left\|{y}_p-{y}_j\right\|_2=\frac{c\pi}{\check \Omega_{\mathrm{multi}, \frac{((n-1)!)^2\sigma}{(2n-1)! m_{\min}}}}\left(\frac{\sigma}{m_{\min}b_{\mathrm{upper}}}\right)^{\frac{1}{2 n-2}} \text{ with } c>1. 
  \]
  In particular, for the case when $n=2$, the lower bound in the above form should be $$\frac{\sqrt{2}\pi}{\check \Omega_{\mathrm{multi}, \frac{((n-1)!)^2\sigma}{(2n-1)! m_{\min}}}}\left(\frac{\sigma}{m_{\min}b_{\mathrm{upper}}}\right)^{\frac{1}{2 n-2}}.$$ Therefore, if the source is not that sparse, obtaining theoretically a better resolution, i.e., $\text{smaller than} \frac{\pi}{\check \Omega_{\mathrm{multi}, \frac{((n-1)!)^2\sigma}{(2n-1)! m_{\min}}}}$,  is extremely hard. This explains experimentally observed phenomena in BEAM. For example, as shown in Supplementary Figure 1:E in \cite{beam2022}, a threefold resolution improvement is achieved for resolving two positive point sources but the algorithm fails to resolve both the number and location of six positive point sources.

	 \subsection{Resolution limit for imaging with unknown illumination patterns}
	 When the illumination patterns are unknown, although we do not have the explicit form of  $A, A^*$, the results in \cite{liu2022mathematical} in one-dimensional space demonstrated that sparsity recovery in  multi-illumination imaging can also improve the resolution, compared to the case of imaging from a single snaptshot. 
	 
	 	To be more specific, let us suppose that we illuminate the source $\mu= \sum_{j=1}^n a_j \delta_{y_j}$ by $N$ different $I_t, \ t=1, \cdots, N,$ and that the measurements are given by
	 	\[
	 	\vect Y(\xi, t) = \sum_{j=1}^n a_j e^{iy_j\xi}, \quad \xi \in [-\Omega, \Omega],
	 	\] 
	 	with $\Omega$ being the cutoff frequency of the imaging system. In \cite{liu2022mathematical}, the authors have demonstrated that,  when 
	 \begin{equation}\label{equ:sparsityresolutionlimit2}
	 	\min_{p\neq j}\babs{y_p-y_j}\geq \frac{2.2e\pi }{\Omega }\Big(\frac{1}{\sigma_{\infty, \min}(IM)}\frac{\sigma}{m_{\min}}\Big)^{\frac{1}{n}},
	 \end{equation}
	 the sparsity-promoting approach ($l_0$ minimization) can stably recover the source locations $y_j$'s. Here, $\Omega$ is the cutoff frequency of the imaging system rather than the $\Omega_{\mathrm{multi}}$ from the $PSF_{\mathrm{multi}}$.  $IM$ is the illumination matrix defined as 
	 \[
	 IM=\begin{pmatrix}
	 	I_{1}(y_{1}) & \cdots & I_{1}(y_{n}) \\ 
	 	\vdots & \ddots & \vdots \\ 
	 	I_{N}(y_{1}) & \cdots & I_{N}(y_{n})
	 \end{pmatrix}
	 \]
	 and $\sigma_{\infty, \min}(IM):=\min _{x \in \mathbb{C}^{k},\|x\|_{\infty} \geq 1}\|IM x\|_{\infty}$ characterizes the incoherence between the illumination patterns. Compared to the estimate of the resolution limit in the case of a single measurement \cite{liu2021theorylse, liu2022mathematicalSR}, which is of order $O\left(\frac{\pi}{\Omega }\Big(\frac{\sigma}{m_{\min}}\Big)^{\frac{1}{2n-1}}\right)$,  multi-illumination imaging will certainly improve the resolution when the incoherence between the illumination patterns is sufficiently high. 
	 
	On the other hand, by the following proposition (\cite[Proposition 2.1]{liu2022mathematical}), it was shown that, in the worst-case scenario, the resolution order $O\left(\frac{1}{\Omega} \left(\frac{\sigma}{m_{\min}}\right)\right)^{\frac{1}{n}}$ is the best that can be obtained if the illumination patterns are completely unknown.

	 \begin{prop} 
	 Given $n \geq 2, \sigma, m_{\min }$ with $\frac{\sigma}{m_{\min }} \leq 1$, and an unknown illumination pattern $I_t$ with $\left|I_t(y)\right| \leq 1, y \in \mathbb{R}, 1 \leq t \leq N$, let $\tau$ be defined by
	 $$
	 \tau=\frac{0.043}{\Omega}\left(\frac{\sigma}{m_{\min }}\right)^{\frac{1}{n}} .
	 $$
	 Then there exist $\mu=\sum_{j=1}^n a_j \delta_{y_j}$ with $n$ supports at $\{-\tau,-2 \tau, \ldots,-n \tau\}$ and $\left|a_j\right|=m_{\min }, 1 \leq j \leq$ $n$, and $\rho=\sum_{j=1}^{n} \hat{a}_j \delta_{\hat{y}_j}$ with $n$ supports at $\{0, \tau, \cdots,(n-1) \tau\}$, such that
	 \[
	 \text{there exist $\hat{I}_t$ 's so that $\frac{1}{2 \Omega} \int_{-\Omega}^{\Omega}\babs{\mathcal{F}\left[\hat{I}_t \rho\right](\xi)-\mathcal{F}\left[I_t \mu\right](\xi)}d\xi<\sigma, t=1, \cdots, N$.}
	 \]
 \end{prop}

Note that all of these results show that it is possible to achieve better resolution than predicted by the operator theory of Section \ref{section:convolutheory1} for recovering very sparse sources, but it is difficult to recover more point sources.

		\section{Conclusions}
		In this paper, we have derived a stability analysis for reconstructing the frequency information of a source and demonstrated that the resolution of multi-illumination imaging is fundamentally determined by the essential bandwidth (or cutoff frequency) of an imaging kernel (or point spread function) formulated in terms of the illumination patterns and  point spread function of the imaging system. Our theory provides a unified way to estimate the resolution of various existing super-resolution modalities and arrive at the same results as those obtained experimentally. Our theory also allows us to estimate the resolution of sparsity-based multi-illumination imaging. In particular, we have shown that  sparsity-promoting algorithms can achieve better resolution than that predicted by operator theory in multi-illumination imaging, provided that the source to be recovered is very sparse.

		\appendix

\section{Proofs of results in Section \ref{section:rslforsparsesr}}\label{section:proofofsparsityresolution}
We first introduce some notation and lemmas that are used in the following proofs. Set
\begin{equation}\label{equ:defiofphi}
	\phi_s(t)=\left(1, t, \cdots, t^s\right)^{\top},
\end{equation}
where the superscript $\top$ denotes the transpose. 
We recall the Stirling formula 
\begin{equation}\label{equ:stirlingformula}
	\sqrt{2 \pi} n^{n+\frac{1}{2}} e^{-n} \leq n ! \leq e n^{n+\frac{1}{2}} e^{-n} .
\end{equation}
We introduce the following useful lemma (\cite[Lemma 5]{liu2021theorylse}). 
\begin{lem}\label{lem:vandermonde1}
	Let $t_1, \cdots, t_k$ be $k$ different real numbers and let $t$ be a real number. We have
	\[
	\left(D_k(k-1)^{-1} \phi_{k-1}(t)\right)_j=\Pi_{1 \leq q \leq k, q \neq j} \frac{t-t_q}{t_j-t_q},
	\]
	where $D_k(k-1):=\left(\phi_{k-1}\left(t_1\right), \cdots, \phi_{k-1}\left(t_k\right)\right)$ with $\phi_{k-1}(\cdot)$ being defined by (\ref{equ:defiofphi}).
\end{lem}

\subsection{Proofs of Propositions \ref{prop:resolutionofsparsitymultiillu1} and \ref{prop:positiveresolutionofsparsitymultiillu1}}
\begin{proof}
\textbf{Step 1.} Consider $\gamma=\sum_{j=1}^{2n} a_j \delta_{\mathbf{t}_j}$ with $\mathbf{t}_1=\left(-\left(n-\frac{1}{2}\right) \tau, 0, \ldots, 0\right), \mathbf{t}_2=\left(-\left(n-\frac{3}{2}\right) \tau, 0, \ldots, 0\right), \ldots$ $\mathbf{t}_{2 n}=\left(\left(n-\frac{1}{2}\right) \tau, 0, \ldots, 0\right)$ and 
\begin{equation}\label{equ:proofpropsparsemultirsl-3}
	\tau=\frac{e^{-1}}{\check \Omega_{\mathrm{multi}, \frac{(n-1)!n!\sigma}{(2n)! m_{\min}}}} \left(\frac{\sigma}{m_{\min}b_{\mathrm{upper}}}\right)^{\frac{1}{2 n-1}}.
\end{equation}
For every ${\xi}=\left(\xi_1, \xi_2, \ldots, \xi_d\right)^{\top}$, $\mathcal{F} [\gamma]({\xi})=\sum_{j=1}^{2 n} a_j \mathrm{e}^{\mathrm{i} \mathbf{t}_j \cdot \xi}=\sum_{j=1}^{2 n} a_j \mathrm{e}^{\mathrm{i}(-n-\frac{1}{2}+j) \tau \xi_1}$. This reduces the estimation of $\mathcal{F} [\gamma]({\xi})$ to the one-dimensional case. In what follows, we demonstrate that with proper $a_j$'s, we have 
\[
\babs{\sum_{j=1}^{2 n} a_j \mathrm{e}^{\mathrm{i}(-n-\frac{1}{2}+j) \tau \xi_1}}< \sigma,  \quad  |\xi_1|\leq {\check \Omega_{\mathrm{multi}, \frac{(n-1)!n!\sigma}{(2n)! m_{\min}}}}.
\]

Let 
\begin{equation}\label{equ:proofpropsparsemultirsl-1}
	t_1=-\left(n-\frac{1}{2}\right) \tau,\ t_2=-\left(n-\frac{3}{2}\right) \tau,\ \cdots,\ t_n=-\frac{\tau}{2}, t_{n+1}=\frac{\tau}{2},\ \cdots,\ t_{2 n}=\left(n-\frac{1}{2}\right) \tau.
\end{equation} 
Consider the following system of linear equations:
\begin{equation}\label{equ:proofpropsparsemultirsl1}
A a=0, 
\end{equation}
where $A=\left(\phi_{2 n-2}\left(t_1\right), \cdots, \phi_{2 n-2}\left(t_{2 n}\right)\right)$ with $\phi_{2 n-2}(\cdot)$ being defined by (\ref{equ:defiofphi}). Since $A$ is underdetermined, there exists a nontrivial solution $a=\left(a_1, \cdots, a_{2 n}\right)^{\top}$. By the linear independence of any $(2 n-1)$ column vectors of $A$, all the $a_j$ 's are nonzero. By a scaling of $a$, we can assume that
\begin{equation}\label{equ:proofpropsparsemultirsl-2}
\min _{1 \leq j \leq n}\left|a_{j}\right|=m_{\text {min }} .
\end{equation}
We define
\[
\mu=\sum_{j=1}^n a_{j} \delta_{t_{j}}, \quad \hat{\mu}=\sum_{j=n+1}^n-a_{j} \delta_{t_{j}} .
\]
We now prove that
\[
\babs{\mathcal{F}[\hat{\mu}](\xi_1)-\mathcal{F}[\mu](\xi_1)}<\frac{\sigma}{b_{\mathrm{upper}}},\quad  |\xi_1|\leq {\check \Omega_{\mathrm{multi}, \frac{(n-1)!n!\sigma}{(2n)! m_{\min}}}}.
\]

\textbf{Step 2.} We first estimate $\sum_{j=1}^{2n}\left|a_j\right|$. We begin by ordering the $a_j$ 's such that
\[
\left|a_{j_1}\right| \leq\left|a_{j_2}\right| \leq \cdots \leq\left|a_{j_{2n}}\right|.
\]
Note that $\left|a_{j_1}\right| \leq m_{\min}$ by (\ref{equ:proofpropsparsemultirsl-2}). Then (\ref{equ:proofpropsparsemultirsl1}) implies that
\[
a_{j_1} \phi_{2 n-2}\left(t_{j_1}\right)=\left(\phi_{2 n-2}\left(t_{j_2}\right), \cdots, \phi_{2 n-2}\left(t_{j_{2 n}}\right)\right)\left(-a_{j_2}, \cdots,-a_{j_{2 n}}\right)^{\top},
\]
and hence
$$
a_{j_1}\left(\phi_{2 n-2}\left(t_{j_2}\right), \cdots, \phi_{2 n-2}\left(t_{j_{2 n}}\right)\right)^{-1} \phi_{2 n-2}\left(t_{j_1}\right)=\left(-a_{j_2}, \cdots,-a_{j_{2 n}}\right)^{\top} .
$$
Together with Lemma \ref{lem:vandermonde1}, we have
$$
a_{j_1} \Pi_{2 \leq q \leq 2 n-1} \frac{t_{j_1}-t_{j_q}}{t_{j_{2 n}}-t_{j_q}}=-a_{j_{2 n}} .
$$
Furthermore,
$$
\begin{aligned}
	\left|a_{j_{2 n}}\right| & =\left|a_{j_1}\right| \Pi_{2 \leq q \leq 2 n-1} \frac{\left|t_{j_1}-t_{j_q}\right|}{\left|t_{j_{2 n}}-t_{j_q}\right|}=\left|a_{j_1}\right| \Pi_{2 \leq q \leq 2 n-1} \frac{\left|t_{j_1}-t_{j_q}\right|}{\left|t_{j_{2 n}}-t_{j_q}\right|} \frac{\left|t_{j_1}-t_{j_{2 n}}\right|}{\left|t_{j_{2 n}}-t_{j_1}\right|} \\
	& =\left|a_{j_1}\right| \frac{\Pi_{2 \leq q \leq 2 n}\left|t_{j_1}-t_{j_q}\right|}{\Pi_{1 \leq q \leq 2 n-1}\left|t_{j_{2 n}}-t_{j_q}\right|} \leq\left|a_{j_1}\right| \frac{\max _{j_1=1, \cdots, 2 n} \Pi_{2 \leq q \leq 2 n}\left|t_{j_1}-t_{j_q}\right|}{\min _{j_{2 n}=1, \cdots, 2 n} \Pi_{1 \leq q \leq 2 n-1}\left|t_{j_{2 n}}-t_{j_q}\right|} .
\end{aligned}
$$
Thus, based on the distribution of $t_j$ 's (\ref{equ:proofpropsparsemultirsl-1}), we have
$$
\left|a_{j_{2 n}}\right| \leq \frac{(2 n-1) !}{(n-1) !n!}\left|a_{j_1}\right| \leq \frac{(2 n-1) !}{(n-1) !n!} m_{\min },
$$
and consequently,
$$
\sum_{j=1}^{2 n}\left|a_j\right|=\sum_{q=1}^{2 n}\left|a_{j_q}\right| \leq(2 n)\left|a_{j_{2 n}}\right| \leq \frac{(2 n)!}{(n-1) !n!} m_{\min } .
$$
It then follows that for $k \geq 2 n-1$,
$$
\left|\sum_{j=1}^{2 n} a_j t_j^k\right| \leq \sum_{j=1}^{2 n}\left|a_j\right|\left((n-\frac{1}{2}) \tau\right)^k \leq  \frac{(2 n)!}{(n-1) !n!} m_{\min }\left((n-\frac{1}{2}) \tau\right)^k.
$$
\textbf{Step 3.} On the other hand, we can expand $\mathcal{F}[\mu-\hat \mu]$ as follows:
\[
\mathcal{F}[\mu-\hat \mu](x)=\sum_{j=1}^{2 n} a_j e^{i t_j x}=\sum_{j=1}^{2 n} a_j \sum_{k=0}^{\infty} \frac{\left(i t_j x\right)^k}{k !}=\sum_{k=0}^{\infty} Q_k(\mu-\hat \mu) \frac{(i x)^k}{k !},
\]
where $Q_k(\gamma)=\sum_{j=1}^{2 n-1} a_j t_j^k$. Based on the discussions in Step 1 and Step 2, we have 
\[
Q_k(\gamma)=0, k=0, \cdots, 2 n-2, \text { and }\left|Q_k(\gamma)\right| \leq \frac{(2 n) !}{n !(n-1) !} m_{\min }((n-1 / 2) \tau)^k, k \geq 2 n-1.
\]
Therefore, for $|x| \leq \check \Omega_{\mathrm{multi}, \frac{(n-1)!n!\sigma}{(2n)! m_{\min}}}$,
\[
\begin{aligned}
	& \max _{x \in\left[-\check \Omega_{\mathrm{multi}, \frac{(n-1)!n!\sigma}{(2n)! m_{\min}}},\ \check \Omega_{\mathrm{multi}, \frac{(n-1)!n!\sigma}{(2n)! m_{\min}}}\right]}\babs{\mathcal{F}\left[\mu - \hat \mu\right](x)}\\
	 \leq& \sum_{k \geq 2 n-1} \frac{(2 n) !}{n !(n-1) !} m_{\min }((n-1 / 2) \tau)^k \frac{|x|^k}{k !} \\
	\leq& \sum_{k \geq 2 n-1} \frac{(2 n) !}{n !(n-1) !} m_{\min }((n-1 / 2) \tau)^k \frac{\check \Omega_{\mathrm{multi}, \frac{(n-1)!n!\sigma}{(2n)! m_{\min}}}^k}{k !} \\
	 =&\frac{(2 n) ! m_{\min }(n-1 / 2)^{2 n-1}\left(\tau \check \Omega_{\mathrm{multi}, \frac{(n-1)!n!\sigma}{(2n)! m_{\min}}}\right)^{2 n-1}}{n !(n-1) !(2 n-1) !} \sum_{k=0}^{+\infty} \frac{\left(\tau \check \Omega_{\mathrm{multi}, \frac{(n-1)!n!\sigma}{(2n)! m_{\min}}}\right)^k(2 n-1) !(n-1 / 2)^k}{(k+2 n-1) !}, 
	\end{aligned}
\]
	and hence 
 \[	\begin{aligned}
		& \max _{x \in\left[-\check \Omega_{\mathrm{multi}, \frac{(n-1)!n!\sigma}{(2n)! m_{\min}}},\ \check \Omega_{\mathrm{multi}, \frac{(n-1)!n!\sigma}{(2n)! m_{\min}}}\right]}\babs{\mathcal{F}\left[\mu - \hat \mu\right](x)}\\ 
	 <&\frac{2 n m_{\min }(n-1 / 2)^{2 n-1}\left(\tau \check \Omega_{\mathrm{multi}, \frac{(n-1)!n!\sigma}{(2n)! m_{\min}}}\right)^{2 n-1}}{n !(n-1) !} \sum_{k=0}^{+\infty}\left(\frac{\tau \check \Omega_{\mathrm{multi}, \frac{(n-1)!n!\sigma}{(2n)! m_{\min}}}}{2}\right)^k \\
	=&\frac{2 n m_{\min }(n-1 / 2)^{2 n-1}\left(\tau \check \Omega_{\mathrm{multi}, \frac{(n-1)!n!\sigma}{(2n)! m_{\min}}}\right)^{2 n-1}}{n !(n-1) !} \frac{1}{0.8} \quad\left(\text { (\ref{equ:proofpropsparsemultirsl-3}) implies } \frac{\tau \check \Omega_{\mathrm{multi}, \frac{(n-1)!n!\sigma}{(2n)! m_{\min}}}}{2} \leq 0.2\right) \\
	\leq &\frac{n m_{\min }(n-1 / 2)^{2 n-1}}{\pi n^{n+\frac{1}{2}}(n-1)^{n-\frac{1}{2}}}\left(e \tau \check \Omega_{\mathrm{multi}, \frac{(n-1)!n!\sigma}{(2n)! m_{\min}}}\right)^{2 n-1} \frac{1}{0.8} \quad(\text { by (\ref{equ:stirlingformula})}) \\
	\leq& \frac{n}{\pi(n-1 / 2)} m_{\min }\left(e \tau \check \Omega_{\mathrm{multi}, \frac{(n-1)!n!\sigma}{(2n)! m_{\min}}}\right)^{2 n-1} \frac{1}{0.8} \\
	<&\frac{\sigma}{b_{\mathrm{upper}}}.  \quad\left(\text { by (\ref{equ:proofpropsparsemultirsl-3}) and } \frac{n}{\pi(n-1 / 2)} \frac{1}{0.8}<1\right).
\end{aligned}
\]
It then follows that 
\[
\babs{\mathcal{F}[\hat{\mu}](\xi_1)-\mathcal{F}[\mu](\xi_1)}<\frac{\sigma}{b_{\mathrm{upper}}},\quad  |\xi_1|\leq \check \Omega_{\mathrm{multi}, \frac{(n-1)!n!\sigma}{(2n)! m_{\min}}}.
\]
By Step 1, we construct the sources in $\mathbb R^d$ as follows:
\[
\mu=\sum_{j=1}^n a_{j} \delta_{\vect t_{j}}, \quad \hat{\mu}=\sum_{j=n+1}^n-a_{j} \delta_{\vect t_{j}},
\]
with the $a_j$'s satisfying (\ref{equ:proofpropsparsemultirsl1}). Above discussions yield 
\[
\babs{\mathcal{F}[\hat{\mu}](\xi)-\mathcal{F}[\mu](\xi)}<\frac{\sigma}{b_{\mathrm{upper}}},\quad  \bnorm{\xi}_2\leq {\check \Omega_{\mathrm{multi}, \frac{(n-1)!n!\sigma}{(2n)! m_{\min}}}}.
\]
Thus, 
\[
\babs{\mathcal F [PSF_{\mathrm{multi}}](\xi)\mathcal{F}[\hat{\mu}](\xi)-\mathcal F [PSF_{\mathrm{multi}}](\xi)\mathcal{F}[\mu](\xi)}<\sigma,\quad  \bnorm{\xi}_2\leq {\check \Omega_{\mathrm{multi}, \frac{(n-1)!n!\sigma}{(2n)! m_{\min}}}},
\]
by the definition of $b_{\mathrm{upper}}$.

\textbf{Step 4.}
Note that by Step 3, for $\bnorm{\xi}_2> {\check \Omega_{\mathrm{multi}, \frac{(n-1)!n!\sigma}{(2n)! m_{\min}}}},$
\[
\babs{\mathcal{F}[\hat{\mu}](\xi)-\mathcal{F}[\mu](\xi)}= \babs{\sum_{j=1}^{2 n} a_j e^{i t_j \xi_1}}\leq \sum_{j=1}^n \babs{a_j} = \frac{(2 n)!}{(n-1) !n!} m_{\min}.
\] 
By the definition of $\check \Omega_{\mathrm{multi}, \frac{(n-1)!n!\sigma}{(2n)! m_{\min}}}$, for $\bnorm{\xi}_2> {\check \Omega_{\mathrm{multi}, \frac{(n-1)!n!\sigma}{(2n)! m_{\min}}}}$, we also have 
\[
\babs{\mathcal F [PSF_{\mathrm{multi}}](\xi)\mathcal{F}[\hat{\mu}](\xi)-\mathcal F [PSF_{\mathrm{multi}}](\xi)\mathcal{F}[\mu](\xi)}<\sigma.
\]
This completes the proof of Proposition \ref{prop:resolutionofsparsitymultiillu1}.

\textbf{Step 5.} Now we prove Proposition \ref{prop:positiveresolutionofsparsitymultiillu1}. Similar to the proof of Proposition \ref{prop:resolutionofsparsitymultiillu1}, we still only need to consider the one-dimensional case. We define 
\[
\mu=\sum_{j=1}^n a_{2j} \delta_{t_{2j}}, \quad \hat{\mu}=\sum_{j=n+1}^n-a_{2j-1} \delta_{t_{2j-1}}, \quad j =1, \cdots, n,
\]
where the $t_j$'s are defined by (\ref{equ:proofpropsparsemultirsl-1}) and $a=(a_1, \cdots, a_{2n})^{\top}$ satisfies (\ref{equ:proofpropsparsemultirsl1}) and 
\[
a_{2n}>0, \quad \min_{j=1,\cdots, n}|a_{2j}|\geq m_{\min}. 
\]  
Note that the above conditions on $a$ are easy to be satisfied after scaling $a$ in  (\ref{equ:proofpropsparsemultirsl1}). Now we only have to prove that $\mu, \hat \mu$ are positive measures, since the other conclusions of Proposition \ref{prop:positiveresolutionofsparsitymultiillu1} can be shown in the same way as in the previous steps. 

Equation (\ref{equ:proofpropsparsemultirsl1}) implies that
\[
-a_{2 n} \phi_{2 n-2}\left(t_{2 n}\right)=\left(\phi_{2 n-2}\left(t_1\right), \cdots, \phi_{2 n-2}\left(t_{2 n-1}\right)\right)\left(a_1, \cdots, a_{2 n-1}\right)^{\top},
\]
and hence,
\[
-a_{2 n}\left(\phi_{2 n-2}\left(t_1\right), \cdots, \phi_{2 n-2}\left(t_{2 n-1}\right)\right)^{-1} \phi_{2 n-2}\left(t_{2 n-1}\right)=\left(a_1, \cdots, a_{2 n-1}\right)^{\top} .
\]
This together with Lemma \ref{lem:vandermonde1} yields
\begin{equation}\label{equ:proofpropsparsemultirsl4}
-a_{2 n} \Pi_{1 \leq q \leq 2 n-1, q \neq j} \frac{t_{2 n}-t_q}{t_j-t_q}=a_j,
\end{equation}
for $j=1, \cdots, 2 n-1$. Observe first that $\Pi_{1 \leq q \leq 2 n-1, q \neq j}\left(t_{2 n}-t_q\right)$ is always positive for $1 \leq j \leq$ $2 n-1$. For $j=2 n-1$, since $$a_{2 n}>0,-a_{2 n} \Pi_{1 \leq q \leq 2 n-1, q \neq 2 n-1}\left(t_{2 n-1}-t_q\right)$$ is negative in (\ref{equ:proofpropsparsemultirsl4}). Thus we have $a_{2 n-1}<0$. In the same way, we see that $a_j>0$ for even $j$ and $a_j<0$ for odd $j$. Therefore, the intensities in $\hat{\mu}$ and $\mu$ are all positive. This completes the proof of Proposition \ref{prop:positiveresolutionofsparsitymultiillu1}. 
\end{proof}

\subsection{Proof of Proposition \ref{prop:positiveresolutionofsparsitymultiillu2}}
\begin{proof}
\textbf{Step 1.} Similar to the proofs of Propositions \ref{prop:resolutionofsparsitymultiillu1} and \ref{prop:positiveresolutionofsparsitymultiillu1}, we only need to consider the one-dimensional case. 

For $j \in\{1,2, \cdots, 2 n\}$, set $t_j=-\frac{s n-2}{2} \tau+\frac{(j-2) s}{2} \tau$ if $j$ is even and $t_j=t_{4\left\lceil\frac{j+1}{4}\right\rceil-2}+$ $(-1)^{\frac{j+1}{2}} \tau$ otherwise. Consider the following system of linear equations:
\[
A a=0,
\]
where $A=\left(\phi_{2 n-2}\left(t_1\right), \cdots, \phi_{2 n-2}\left(t_{2 n}\right)\right)$ with $\phi_{2 n-2}(\cdot)$ defined in (\ref{equ:defiofphi}). As discussed before, for a nonzero $a$, by scaling, we can assume that $a_{2 n}>0$ and
$$
\min _{1 \leq j \leq n}\left|a_{2 j}\right|=m_{\min}.
$$
We define
$$
\mu=\sum_{j=1}^n a_{2 j} \delta_{t_{2 j}}, \quad \hat{\mu}=\sum_{j=1}^n-a_{2 j-1} \delta_{t_{2 j-1}} .
$$
As discussed before, we can show that $a_{2 j-1}<0, j=1, \cdots, n$, and $a_{2 j}>0, j=1, \cdots, n$. Thus, both $\hat{\mu}$ and $\mu$ are positive measures.

\textbf{Step 2.} By the proof of Proposition 2.1 in \cite{liu2022mathematicalpositive} we know that
\[
\sum_{j=1}^{2 n}\left|a_j\right| \leq \sum_{q=1}^{2 n}\left|a_{j_q}\right| \leq \frac{2ne^{11} s^2(n+1)^{10} 2^{2 n-8} m_{\min}}{\pi^2}
\]
and 
\[
\babs{\mathcal F \left [ \mu \right](\xi)- \mathcal F \left [ \hat \mu \right](\xi)}<\sigma, \quad \xi \in \left[- \check \Omega_{\mathrm{multi},  \frac{\pi^2\sigma}{2ne^{11} s^2(n+1)^{10} 2^{2 n-8} m_{\min}}}, \check \Omega_{\mathrm{multi},  \frac{\pi^2\sigma}{2ne^{11} s^2(n+1)^{10} 2^{2 n-8} m_{\min}}} \right].
\]
On the other hand, by the definition of $\check \Omega_{\mathrm{multi},  \frac{\pi^2\sigma}{2ne^{11} s^2(n+1)^{10} 2^{2 n-8} m_{\min}}}$, for $\bnorm{\xi}_2> \check \Omega_{\mathrm{multi},  \frac{\pi^2\sigma}{2ne^{11} s^2(n+1)^{10} 2^{2 n-8} m_{\min}}}$, we also have 
\[
\babs{\mathcal F [PSF_{\mathrm{multi}}](\xi)\mathcal{F}[\hat{\mu}](\xi)-\mathcal F [PSF_{\mathrm{multi}}](\xi)\mathcal{F}[\mu](\xi)}<\sigma.
\]
This completes the proof.

\subsection{Proof of Proposition \ref{prop:numberresolutionofsparsitymultiillu1}}
Let
\[
\tau=\frac{e^{-1}}{\check \Omega_{\mathrm{multi}, \frac{((n-1)!)^2\sigma}{(2n-1)!m_{\min}}}}\left(\frac{\sigma}{m_{\min } b_{\mathrm{upper}}}\right)^{\frac{1}{2 n-2}}
\]
and $\vect t_1=\left(-(n-1) \tau,0, \cdots, 0 \right), \vect t_2=\left(-(n-2) \tau,0, \cdots, 0\right) \cdots, \vect t_{2 n-1}=\left((n-1) \tau, 0, \cdots, 0\right)$. Define
\[
\mu=\sum_{j=1}^n a_{2 j-1} \delta_{\vect t_{2 j-1}}, \quad \hat{\mu}=\sum_{j=1}^{n-1}-a_{2 j} \delta_{\vect t_{2 j}} .
\]
The rest of proof proceeds in the same way as the proofs of Propositions \ref{prop:resolutionofsparsitymultiillu1} and \ref{prop:positiveresolutionofsparsitymultiillu1}. 
\end{proof}

		\bibliographystyle{plain}
		\bibliography{reference_final}	
	\end{document}